\documentclass[11pt]{article}

\usepackage{epsfig}
\usepackage[section]{algorithm}
\usepackage{algorithmic}

\usepackage{anysize}
\usepackage{amsmath}
\usepackage{amssymb}
\marginsize{1in}{1in}{1in}{1in}

\def\qed{\hbox{\rlap{$\sqcap$}$\sqcup$}}
\newenvironment{proof}{\par\noindent{\bf Proof:}}{\mbox{}\hfill$\qed$\\}
\newtheorem{theorem}{Theorem}[section]
\newtheorem{lemma}{Lemma}[section]

\newcounter{rem}
\newcommand{\ignore}[1]{ }

\begin{document}

\title{The Budgeted Transportation Problem}
\author{S. Kapoor \thanks{Computer Science, IIT,Chicago,IL, e-mail: kapoor@iit.edu} and M. Sarwat\thanks{Amazon.com, Work done at Computer Science, IIT, Chicago, IL, e-mail sarwat@iit.edu}}
%\degree{Doctor of Philosophy}
%\dept{Computer Science}
\date{}
\maketitle

\begin{abstract}
Consider a transportation problem with sets of sources and sinks. There
are profits and prices on the edges. The goal is to maximize the profit while meeting
the following constraints; the total flow going out of a source must not exceed its
capacity and the total price of the incoming flow on a sink must not exceed its budget.
This problem is closely related to the generalized flow problem.

%and recently addressed
%auction model of Ad-words.

We propose an auction based primal dual approximation algorithm to solve the problem.
The complexity is $O(\epsilon^{-1}(n^2+ n\log{m})m\log U)$ where $n$ is the number
of sources, $m$ is the number of sinks, $U$ is the ratio of the
maximum profit/price  to the minimum profit/price.

We also show how to generalize the scheme to solve a more general version of the
problem, where there are edge capacities and/or the profit function is concave and 
piecewise linear. The complexity of the algorithm depends on  the
number of linear segments, termed ${\cal L}$, of the profit function.
%remains the same, $O(\epsilon^{-1}(n^2+ n\log{m})m\log U)$.

\end{abstract}

\section{Introduction}
\label{chap_BT}

%\Section{The Problem}

%\label{section_background}
%Keyword based advertising is widely applied by internet based service companies like Google and
%Yahoo. Consider the following model for assigning advertisement space to the bidders. Each bidder
%specifies the its total budget, the font size, list of adwords and price per line. Each 
%advertisement space (associated with a query or email) has limited length, and associated 
%adwords. A bidder is interested if the query contains one or more of its adwords. The goal
%is to maximize the revenue without exceeding the bidders budget or advertising space.

%This problem is readily modelled by the Maximum Profit Budgeted transportation problem. The
%bidders are the sinks, their budget is the capacity. Font size (line height) is the   
The Transportation Problem is a fundamental problem in Computer Science. It was 
initially formulated to represent the problem of transporting goods from warehouses to customers.
The mincost version of the problem is also known as the Hitchcock problem after one of its
early formulators \cite{hitchcock41dist}. It is well-known
that while the Transportation problem
can be seen as a special case of the mincost flow problem, the latter can also be transformed 
into the former \cite{papadimitriou}.

The first primal-dual algorithm, called the Hungarian method \cite{Kuhn55Hungarian}, 
was also proposed for 
the Assignment Problem which happens to be a 
special case of the Transportation Problem. More
recently, the Transportation Problem was used by 
Bertsekas \cite{bertsekastutorial} and 
Goldberg and Tarjan  \cite{goldberg87solving} as a natural framework for applying 
auction based algorithms. 

The typical transportation model assumes that there is no loss of goods as the goods are 
transported from the sources to the sinks. In real life, there is often a loss involved. 
For example, in case of  electricity there is loss due to resistance and in case of oil, 
due to evaporation. In case of a delicate merchandise, there might breakage. 
Sometimes, there may even be a gain (e.g. transfer of money: currency conversion).
The loss or gain can be modeled as a price function on the transport from the 
source to the sink.

As an example, consider  a set of depots and a set of retail stations. 
The depots have a certain supply of goods, 
while the retail stations have a limit on their intake capacity.
For every unit transported between a pair of depot and station, 
a certain amount of profit is made
which depends on selling price, cost of transportation etc. 
In addition, consider the situation where goods transported 
suffer a loss 
proportional to the number of units transported.
This factor is unique to each pair as it depends on
the mode of transfer and distance etc. 
In order to measure the net incoming amount at
the stations, we have to account for this loss. 
The intake capacity limit is therefore, on a weighted
sum rather than a simple sum.    

In a more recent context, keyword based advertising is widely applied by 
internet based service companies like Google and
Yahoo. Consider the following model for assigning advertisement space to the 
bidders: each bidder, say $j$, modeled by a sink,
specifies its total budget $b_j$ and the price $p_{ij}$ it is willing to pay 
for the $i$th keyword, specified by  source $i$.
An estimate of the number of times a keyword would be invoked in web 
searches provides a capacity $a_i$ for the $i$th keyword.
Assignment of a keyword $i$ to a bidder $j$ provides a profit $c_{ij}$.
%The number of times a particular key-word is assigned to a bidder
%, the
%, the font size, list of adwords and price per line. Each 
%advertisement space (associated with a query or email) has limited length, and associated 
%adwords. A bidder is interested if the query contains one or more of its adwords. 
The goal is to maximize the revenue of the service provider without exceeding the bidders budget.
%This problem is readily modelled by the Maximum Profit Budgeted transportation problem. 
A version of the problem restricted to an integral setting may be 
found in \cite{andelman04auctions}.
The online version of the problem has been considered in \cite{mehta05adwords}.

In order to model the above and similar problems, we propose a generalization of the 
transportation problem: the Budgeted  Transportation Problem (BTP) which is closely related 
to Generalized Flow Problem. We then present an auction based approximation algorithm. 
We also show how to extend the technique to further generalizations of the problem. 

Like the conventional Transportation Problem, the problem is 
modeled  on a weighted bipartite graph. 
In our version of the problem we consider maximizing the profit in 
transporting/assigning goods
from the sources to the sinks, subject to the following conditions:
Each source, $i$  has a {\it supply} of goods, bounded by a capacity function $a_i$,
and each  sink $j$ has a bound on the incoming set of goods, 
specified by a  {\it budget} $b_j$.
The difference from the conventional transportation problem is that the  budget is a 
bound on the weighted sum of the incomming goods as follows: each  source-sink 
pair $(i,j)$ has a {\it price} function $p_{ij}$ and the bound $b_j$ is on the 
sum $ \sum_i p_{ij}f_{ij}$ where $f_{ij}$ is the amount of goods transported/assigned 
from source $i$ to sink  $j$. 
The goal is to maximize the sum of profits of transporting/assigning goods 
from sources to sinks,
given that the profit obtained by 
transporting/assigning $f_{ij}$ units is $c_{ij}f_{ij}$.
%the weighs specified by a price function
%\where sources have a capacity and the sinks have demands. The total incoming flow
%may not exceed the demand. In BTP we associate a {\em price} with every edge. We have  
%a limitation on the total price on a sink, which is the weighted sum of the 
%incoming flow. This limit is called the {\em budget} of the sink. We thus have two
%weight function on the edges with a budget based on one of them. The total flow going out of 
%a source is limited by its capacity like the conventional Transportation problem. 
The problem can be formally stated as follows: 

Given a bi-partite graph $G(S,T)$, capacity $a: S \rightarrow \mathbb{Z}^+$, 
budget $b: T \rightarrow \mathbb{Z}^+$,
price $p: S \times T \rightarrow \mathbb{Z}^+$, profit $c: S \times T \rightarrow \mathbb{Z}^+$ the budgeted transportation
problem is to find a flow function $f: S \times T \rightarrow \mathbb{R}^+$
so as to maximize $\sum_{i \in S, j \in T} c_{ij}f_{ij}$ subject to capacity constraints, $\sum_{j \in T} f_{ij} \leq a_i$
and budget constraints $\sum_{i \in S} p_{ij}f_{ij} \leq b_j , \forall j \in T$.We let $|S| = n$ and $|T| = m$.

Note that, for simplicity, the price function, like the other functions,
is assumed to be integral in the 
paper. This could be replaced by a  price function to the set of rationals.

%\begin{figure}[h]
%\centerline{\epsfysize=200pt\epsfbox{bt.eps}}
%\caption{A transportation setup with losses} 
%\end{figure}

We describe an approximation scheme for the budgeted transportation problem.
Our algorithm is a primal-dual auction algorithm.
Auction algorithms have been utilized before in the solution of combinatorial
problems \cite{bertsekastutorial,goldberg87solving} including the classical transportation 
problem \cite{bertsekas89transport}, and more recently for finding market 
equilibrium \cite{garg04auction, garg04auctionbased}. 
Bertsekas et al \cite{bertsekas97anerelaxation} apply the auction mechanism  to 
solve the generalized flow problem but the complexity of that 
algorithm is not polynomial.

While the auction technique has been used before, the application in the current context
is different.  The auction mechanism is used to realize a set of tight edges, i.e. edges 
which satisfy complementary slackness conditions. Paths and cycle are found in a subgraph 
in these set of tight edges. The flow path/cycle could either be flow generating or 
create a flow reduction. The primal-dual nature of our approach allows us to push flow 
along these paths or cycle without worrying about the profit of the path/cycle. The cycle 
can be eliminated in linear time. Our algorithm may be interpreted to have achieved a  
version of cycle-canceling \cite{wayne99apolynomial} without having to compute 
the associated profit.

We achieve a complexity of $O(\epsilon^{-1}(n^2+ n\log{m})m\log U\log{m})$, 
where $n$ is the number of sources, $m$ is the number of sinks and $U$ is
$\frac{ \max_{ij}(\frac{p_{ij}}{c_{ij}}) }{ \epsilon \min_{ij} (\frac{p_{ij}}{c_{ij}})}$.
The problem can also be modeled as a generalized flow problem.
This can be done by adding a super source connected to all the sources and a supersink
connected to all the sinks and minimizing the negative of the profit function. 
By using the technique of scaling as in \cite{Luby-Nisan},
the dependancy on $\log U$ can be replaced by  $\log (nm/\epsilon)$.
We will thus choose to ignore terms involving $\log U$, and in fact
$\log n$, when comparing our results with other results.

There are numerous results 
for maximum/mincost generalized flow and related problems 
\cite{tardos98simple,fleischer99fast,wayne99apolynomial,radzik98faster,oldham99combinatorial,
goldberg91combinatorial,goldfarb96afaster}. 
Some of the currently best known FPTAS for the generalized flow problems
are of complexity $\tilde{O}(\log\epsilon^{-1}E(E+V\log{I})))$ 
\cite{fleischer99fast,tardos98simple,radzik98faster} 
for 
maximum  generalized flow and $\tilde{O}(\epsilon^{-2}E^2VJ)$ \cite{fleischer99fast} and 
$\tilde{O}(\log{\epsilon^{-1}}E^2V^2 )$ \cite{wayne99apolynomial} for minimum 
cost generalized flow. Here, $I=\log M$, $J=\log\log M + \log \epsilon^{-1}$ and $M$ is 
the largest integer in cost representation. For the special case, when 
there are no flow generating cycles, the complexities are 
$\tilde{O}(\epsilon^{-2}E^2)$ and $\tilde{O}(\epsilon^{-2}E^2J)$
for maximum and minimum cost generalized flows respectively \cite{fleischer99fast}, where
$V$ is the number of nodes and $E$ is the number of edges in the given graph.

A straight-forward application of the best known generalized flow approximation
algorithm to this problem would require $\tilde{O}(\log{\epsilon^{-1}}(mn)^2(m+n)^2)$ time
using the minimum cost generalized flow algorithm in \cite{wayne99apolynomial}  
or $\tilde{O}(\epsilon^{-2}(mn)^2(m+n)J) $ time using the algorithm in \cite{fleischer99fast}. 
Our algorithm has a better complexity than the first when
$\epsilon > 1/m^3$ and is better than the second for all ranges of $\epsilon$.

If we use the packing algorithm from \cite{garg98faster} than 
one can solve this problem in
$\tilde{O}(\epsilon^{-2}(n+m)nm)$ since we have $m+n$ rows and $mn$ columns
in the primal constraint matrix. Other known combinatorial
algorithms for packing \cite{PST95,GK94}
are also all dependent on $1/\epsilon^2$.  
After the completion of this work we became aware of 
another result \cite{kouf-young}
which provides a randomized $(1+\epsilon)$ approximate solution to the 
packing problem in time $O(N +(r+c)\log (N)/\epsilon^2)$ where
$N$ is the number of non-zero entries in the specification of
the packing program and $r+c$ is the number of rows and columns in the
matrix specifying the packing program. As applied to the budgeted
transportation problem this algorithm would 
require $O(nm ( \log nm)/\epsilon^2)$ and  would be faster
than our approach when $\epsilon \geq 1/n$.  

Comparing with algorithms that have a dependence on $1/\epsilon$,
one recent algorithm for packing problems
is by Bienstock and Iyengar \cite{bienstock:2006}, 
which reduces the packing problem to solve a sequence of 
quadratic programming which in turn are approximated by linear programs.
Consequently, the complexity will have high polynomial dependance on $n$ and
$m$.
We achieve a combinatorial algorithm with 
complexity dependant on $O(1/\epsilon)$ and a low  polynomial
in $n$ and $m$.

\ignore{
Note that the approximation to the budgeted transportation problem does 
not immediately provide an algorithm for the generalized flow problem in general graphs 
since the approximation guarantee does not translate back.
}
We extend our algorithm to solve a generalization of BTP called
the  Budgeted Transshipment 
Problem (BTS). It is similar to BTP except for added edge capacity constraints. 
We achieve the same approximation with no additional time complexity,
i.e $O(\epsilon^{-1}(n^2+ n\log{m})m\log U\log{m})$.
This is of improved efficiency, 
since, again, to compare this with existing results: the result in \cite{wayne99apolynomial} 
would imply a complexity of
$\tilde{O}(\log{\epsilon^{-1}}(mn)^2(m+n)^2)$, \cite{fleischer99fast} 
would imply a complexity of
$\tilde{O}(\epsilon^{-2}(mn)^2(m+n)J)$ while  \cite{garg98faster} would imply
a complexity of $\tilde{O}(\epsilon^{-2}(n+m+mn)nm)$
%\log m)$ 
as the number of constraints is now $m+n+mn$.
The results from \cite{kouf-young} would, again, be of complexity
$O(nm (\log nm)/ \epsilon^2)$.

%We first show the relationship between  BTP and MCGF problems in section \ref{sec_relation}. 
In section \ref{sec_basic}
we present a basic auction based algorithm and show its correctness. In section \ref{sec_modified}
we describe a modified algoritm and discuss its convergence. Since the algorithm is similar to the algorithm for the generalized
problem, BTS, we demonstrate an outline of the algorithm.
Section \ref{sec_extend} extends the result to BTS. In section
\ref{sec_details} we present a detailed algorithm for BTS (which also applies to BTP)
and in subsequent sections the proofs of correctness and complexity of the 
algorithm. The piece-wise linear case is discussed in section~\ref{piece-wise}. 
We conclude in  section \ref{sec_conc5}.
We also  show a relationship between  BTP and  minimum cost generalized flow problems in the appendix (section \ref{sec_relation}).

\section{Auction Based Algorithm}
\label{sec_auc}
\subsection{The Basic Method}
\label{sec_basic}

In this section, we describe a primal-dual auction framework to 
approximate the maximum profit 
budgeted transportation problem (BTP) to within a factor of $(1-\epsilon)$. 
The algorithm is a locally greedy algorithm 
where the sources compete against each other for sinks by bidding with lower 
effective 
profit, a quantity defined using dual variables. 
This algorithm will serve to illustrate the methodology. In order to prove
complexity bounds we will, in the next section, refine the algorithm further.

We first model the problem using an LP and also define its dual.

\begin{eqnarray*}
{\rm maximize\ \ \ \ \ }\sum_{i\in S, j \in T} c_{ij}f_{ij}\\
{\rm subject\ to:}\\ 
\sum_{j \in T} f_{ij} \leq a_i \ \ \ \forall i \in S\\
\sum_{i \in S} p_{ij}f_{ij} \leq b_j \ \ \ \forall j \in T\\
f_{ij} \geq 0 \ \ \ \ \  \forall i \in S ,j \in T
\end{eqnarray*}
In the remainder of the paper, the variables $i$ and $j$ will refer to
a  source $i$ in $S$ and sink $j$ in $T$, respectively.

The dual of the above program is 

\begin{eqnarray*}
{\rm minimize\ \ \ \ \ }\sum_{i \in S } \alpha_i a_i +  \sum_{j \in T} \beta_j b_j\\
{\rm subject\ to:}\\ 
\alpha_i &\geq& c_{ij} - p_{ij}\beta_j \ \ \ \forall ij\\
\alpha_i, \beta_j &\geq& 0
\end{eqnarray*}

\ignore{
The algorithm maintains a pair of feasible dual and primal
solutions. It works to modify them until all the 
complementary slackness conditions are exactly or 
approximately satisfied.
$c_{ij} - p_{ij}\beta_j$
}

%The auction algorithm then picks the uncleared sources in a round-robin fashion,
%and tries to clear them by increasing the valuation (and therefore reducing the
%effective profit).
The following variables based on primal and dual variables are
used in the algorithm.

\begin{itemize}
\item
$\alpha_i$ : the dual variable associated with each source.
\item
$\beta_j$ : the dual variable corresponding to the budget constraint on the sink. 
An algorithmic interpretation of $\beta_j$ is that it represents
the value of sink $j$. It increases as different sources compete in order to ship flow to it.
At every instance of the algorithm, flow may be pushed onto the sink $j$ 
from a source at a given
value of $\beta_j$. We term this an assignment of the source at valuation $\beta_j$.
Further, at any instance of the algorithm sources are assigned to sinks at two prices (values),
some at value $\beta$ while others are assigned at 
a {\bf companion valuation} $\beta'_j$ where 
$\beta_j= \beta'_j(1+\epsilon)$. 

\item
$y_{ij}$:  the {\bf valuation} of sink $j$ when assigned to source 
$i$ to ship the flow from $i$ to $j$. It equals  $\beta_j$ when source $i$ ships
flow to sink $j$.  Subsequently, as $\beta_j$ rises, it becomes equal to $\beta'_j$.
\item
$s_{i}$ : the supply (or surplus) remaining at the source $i$. $s_i$ is  
evaluated as  $s_i = a_i - \sum_j f_{ij}$
 and is implicitly updated in the algorithm.
\item
$d_j$ : the amount remaining from the budget at sink $j$.
$d_j$ is evaluated as $d_j = b_j - \sum_i p_{ij}f_{ij}$.
%\item
%$\beta'_j$ the lowest assigned valuation at the sink $j$.
\end{itemize}

Because the input is asssumed to be integral, we note that all the variables
are rationals, with a bound on the absolute value of the denominators.

The initial and primal dual feasible solutions are obtained by the initialization 
procedure \ref{proc_init} {\bf Initialize}.
The algorithm maintains a pair of feasible dual and primal
solutions. It works to modify them until the following
complementary slackness conditions,
which define optimality, are exactly or approximately satisfied.

\begin{eqnarray}
\label{primal_CScond}
{\rm {\bf CS1:}} &  \ \ f_{ij}(c_{ij} - p_{ij}\beta_j - \alpha_i) = 0 \ \ \ \forall ij\\
\label{source_cond}
{\rm {\bf CS2:}} & \ \ \alpha_i(a_i - \sum_{j}f_{ij}) = 0 \ \ \ \forall i\\
\label{sink_cond}
{\rm {\bf CS3:}} & \ \ \beta_j(b_j - \sum_{i} p_{ij}f_{ij}) = 0 \ \ \ \forall j
\end{eqnarray} 

In each iteration, the algorithm \ref{alg_auc} {\bf Auction} picks a source,
$i$, with  
positive surplus, $s_i= a_i - \sum_j f_{ij}$, 
%greater than $\delta$ 
and $\alpha_i > 0$. It then picks a sink $j$ that provides it the 
largest effective profit, where the profit is specified by the quantity, 
$c_{ij} - p_{ij}\beta_j$. If the sink is not saturated
flow is pushed from source $i$ to sink  $j$.
If sink $j$ is saturated, flow from a 
source  $i'$, that is shipping flow to sink $j$ at a lower of value $\beta'_j$, is replaced
by flow from $i$, i.e. source $i$ wins over the shipment to sink $j$ from source $i'$ by
bidding a higher value ($\beta_j$) for $j$. 
After the change of flow assignment, it is checked  (\ref{proc_simpleupdate} {\bf UpdateBeta( $\beta_j$ ) }) 
that there exists at least one source shipping flow at  the lower value. 
If not, then $\beta_j$ is raised, indicating that
shipping flow to $j$ requires paying a higher value (and lower effective profit).
At this stage, 
the process is repeated considering sinks in order of effective profit until \\
(i) either all the surplus
at node $i$ is pushed out or \\
(ii) $\alpha_i =0$ 
%alternately, 
%(iii) all the sources 
%assigned at the lower value have been replaced.
%In the first two cases the iterations stop. In the third case
%the value $\beta_j$ of sink $j$ is raised in \ref{proc_simpleupdate} {\bf Update $\beta_j$} 
%and the iterations continue.
%The algorithm continues till all the surplus is exhausted %with $\delta$ 
%or when $\alpha_i$ is reduced to zero.

\begin{algorithm}
\begin{algorithmic}[1]
\vspace*{.1in}
%\vspace*{.1in}
\STATE{$\alpha_i \leftarrow \max_j c_{ij}$}
\STATE{$f_{ij} = 0$}
\STATE{$\beta'_j \leftarrow 0$}
\STATE{$\beta_j \leftarrow 0$}
\end{algorithmic}
\caption{{\bf Initialize}}
\label{proc_init}
\end{algorithm}

\begin{algorithm}
\begin{algorithmic}[1]
\vspace*{.1in}
\WHILE{there exists a source $i$ such that $\alpha_i > 0$
       and $\sum_j f_{ij} < a_i $}%- \delta$}

\STATE{Pick a sink $j$ such that $c_{ij}-p_{ij}\beta_j$ is maximized}

\IF{ {\em the sink is saturated, i.e.} $\sum_{i} p_{ij}f_{ij} = b_j$}
\label{alg_auc:line31}

\STATE{Find a source $i'$, currently assigned at the lower value level i.e. $y_{i'j} = \beta'_j$ }
	\IF{$i'\neq i$}
		\STATE{{\em //determine the amount that can be replaced}}
		\STATE{$x \leftarrow \min(s_{i}, f_{i'j}p_{i'j}/p_{ij})$}
		\STATE{{\em // replace the flow}}
		\STATE{$f_{ij} \leftarrow f_{ij}+x$}
		\label{alg_auc:line1}
		\STATE{$ f_{i'j} \leftarrow f_{i'j} - xp_{ij}/p_{i'j}$}
		\label{alg_auc:line2}
		\STATE{$y_{ij} \leftarrow \beta_j$}
	\ELSE
		\STATE{{\em //$i$ is already shipping to $j$ at the lowest level}} 
		\STATE{{\em //Raise the  assigned value without changing any flow} }
		\STATE {$y_{ij} \leftarrow \beta_j$ \label{step_betachange}}
	\ENDIF
	\STATE{Update $\beta_j$ and $\beta'_j$}
	\label{alg_auc:line18}

\ELSE
	\STATE{{\em //The sink is unsaturated, push maximum flow possible }}
	\STATE{ {\em //under supply and demand constraint}}
	\STATE{$x \leftarrow \min (s_i,d_j/p_{ij})$}
	\label{alg_auc:line3}
	\STATE{$f_{ij} \leftarrow f_{ij} + x$}
	\label{alg_auc:line4}

\ENDIF

\STATE {$\alpha_i \leftarrow \max_j (c_{ij}-p_{ij}\beta_j,0) $}
\label{alg_auc:line101}

\IF{$\alpha_i = 0$}
\STATE{$y_{ij} \leftarrow \beta'_j$ for all $j$}
\label{alg_auc:line11}
		\vspace*{-2pt}
\STATE{//This ensures there is no further rise in 
$\beta_j$ without reducing $f_{ij}$ to zero}
\ENDIF

\ENDWHILE

\end{algorithmic}
\caption{ {\bf Auction} }
\label{alg_auc}
\end{algorithm}

\begin{algorithm}
\begin{algorithmic}[1]

\vspace*{.1in}

\IF{$\beta_j$ = 0} 
	\STATE{$\beta_j \leftarrow \epsilon \min_i c_{ij}/p_{ij}$}
\ELSIF {$\forall i:f_{ij} > 0, y_{ij} = \beta_j$}
	\STATE{$\beta'_j \leftarrow \beta_j$}
	\STATE{$\beta_j \leftarrow \beta_j(1+\epsilon)$}
\ELSE
	\STATE{//no update required}
\ENDIF

\end{algorithmic}
\caption { {\bf UpdateBeta} ($\beta_j$ )}
\label{proc_simpleupdate}
\end{algorithm}
To prove the correctness, we show that primal and dual equations are satisfied
and complentary slackness conditions hold. All line numbers in the
discussion below correspond to the description of {\bf Algorithm 2.2  Auction},
except when stated.
\begin{lemma}
The dual solution maintained by the algorithm is always feasible.
\end{lemma}

\begin{proof}
We observe that the value of $\beta_j$ never decreases as the algorithm 
progresses. Therefore, $\alpha_i < c_{ij} - p_{ij}\beta_j$ can only 
become true when $\alpha_i$ is modified. However, when $\alpha_i$ is modified, 
it is set to a value $ \geq \max_j \{ c_{ij} - p_{ij}\beta_j \}$ in line \ref{alg_auc:line101}. 
At that point, $\alpha_i \geq c_{ij} - p_{ij}\beta_j$ for all $j$. 
\end{proof} 
For use in  the next lemma, we define 
the  {\em total price of the incoming flow on sink $j$} to be 
$\sum_{i} p_{ij}f_{ij}$.

\begin{lemma}
The primal solution maintained by the algorithm 
does not exceed the budget constraints. 
\label{primal_feasible1}
\end{lemma}
\begin{proof}
The flow on any edge is increased in lines \ref{alg_auc:line1} and 
\ref{alg_auc:line4}. In line \ref{alg_auc:line1} 
the total price of the incoming flow on sink $j$ increases by $xp_{ij}$. 
It is then reduced by $p_{i'j}x \frac{p_{ij}}{p_{i'j}}$ in line 
\ref{alg_auc:line2}, and therefore, the budget is met. 

In line \ref{alg_auc:line3}, the total price increases by at most 
$d_j = b_{j}-\sum_i p_{ij}f_{ij}$ and thus the budget is not exceeded.
\end{proof}

\begin{lemma}
The primal solution maintained by the algorithm
 does not violate supply constraints.
\label{primal_feasible2}
\end{lemma}
\begin{proof}
During the algorithm, the increase $x$ in flow, $f_{ij}$, on an edge $ij$ 
is always limited by the available supply (line 7 and 21) and
therefore, the supply constraint is clearly satisfied. 
\end{proof}

Combining the above lemmas \ref{primal_feasible1} and \ref{primal_feasible2} 
we conclude:

\begin{lemma}
The primal solution maintained by the algorithm is always feasible.
\end{lemma}

%\begin{lemma}
%The difference between an assigned valuation $y_{ij}$ and $\beta_j$ is no 
%more than than a factor of $(1+\epsilon)$.

%\[
%\forall f_{ij} > 0, y_{ij} \geq \frac{\beta_j}{1+\epsilon}
%\]
%\label{yclosetobeta}

%\end{lemma}

%\begin{proof}
%\end{proof}

\begin{lemma}
\label{beta-change}
During the course of the algorithm, if sink $j$ is not saturated 
then $\beta_j = 0$. Further, $\forall j, \beta_j$ is non-decreasing during the
course of the algorithm.
\end{lemma}
\begin{proof}
The variable $\beta_j$ for all  $j$ is initialized to zero. It is
subsequently changed by the procedure {\bf UpdateBeta}, which in turn is called only
if the sink $j$ is saturated. That is, line \ref{alg_auc:line18} is executed only when the
``{\bf if}'' condition in line \ref{alg_auc:line31} : $\sum_i p_{ij}f_{ij} = b_j$ is 
satisfied. Moreover,  $ \forall j, \beta_j$ only increases in value.
\end{proof}

The following lemma will be useful for bounding the complexity of the various algorithms:
\begin{lemma}
\label{beta-increase}
Suppose $f_{ij}$, the flow on edge $ij$, 
that was pushed at a valuation of sink $j$,  
$y_{ij}$=  $\beta_j'= \beta_j/(1+\epsilon)$, is reduced to zero, 
i.e. pushed back to zero. 
Then at the next occurrence of the event that corresponds to $f_{ij}$ being 
pushed back  on edge $ij$,  
%the valuation of $j$ rises
the valuation of sink $j$, $y_{ij} \geq \beta_j$.
%occurs after $f_{ij}$ is pushed at an increased valuation of 
\end{lemma}
\begin{proof}
This follows from the fact the $\beta_j$ is monotonically non-decreasing.
Flow is pushed back on an edge $ij$
only when there is flow on the edge, $f_{ij}$ 
that had been pushed at the companion valuation 
of $\beta_j/(1+\epsilon)$, where $\beta_j$ is the current valuation.
Let flow be pushed back from sink $j$ to source $i$ on edge $ij$ 
at some time $t$ in the algorithm such that the flow is reduced to zero.
Then when flow is pushed from $i$ to $j$, at a subsequent
time $t' >t$, the  valuation  $y_{ij} \geq \beta_j$ (line 11), 
since the valuation
of $j$ does not decrease. Consequently, if flow is pushed back again
from $j$ to $i$, the valuation has  increased.
\end{proof}

\begin{lemma}
During the execution of the algorithm,
the Complementary Slackness condition, {\bf  CS1 }
is approximately satisfied for each edge $ij$. That is, 
\[
\forall ij:f_{ij} > 0, 
% \ c_{ij}-p_{ij}\beta_j \leq 
\alpha_i \leq c_{ij} - p_{ij}\beta_j +  \epsilon c_{ij}
\] 
\label{lemma_primalCS}
\end{lemma}

\begin{proof}
When flow is pushed on the edge $ij$, $f_{ij} >0$ and 
$\alpha_i = c_{ij}-p_{ij}\beta_j$ or is $0$
(line 2 and 24 in {\bf Algorithm 2.2 Auction}) and the condition is true. 
Note that at this event, $y_{ij}=\beta_j$.
Subsequently, $\alpha_i$ or $\beta_j$ may change.
We consider the changes 
to $\alpha_i$ and $\beta_j$ and 
to the flow, $f_{ij}$, between the event in the algorithm
when positive flow is pushed on the edge $ij$ and the event when
this pushed flow reduces to zero, being pushed back on edge $ij$.
% and is then pushed again on edge $ij$.
In between the two events, since we assume that $f_{ij} >0$,  $y_{ij}$ is either $\beta_j$ or
$\beta_j/(1+\epsilon)$.
We consider cases
depending on whether $\alpha_i >0$ or not.
% computed in line 24 of the algorithm, {\bf Algorithm 2.2 Auction}:
\begin{enumerate}

\item
$\alpha_i >0$. Suppose $\alpha_i$ does not change but $\beta_j$ has
increased by a factor of  $(1+\epsilon)$.
%set to $c_{ij}-p_{ij}\beta_j$ which is positive. 
%In this case, both $\alpha_i$ and $\beta_j$ can change. 
%If $f_{ij}>0$, then either $y_{ij}=\beta_j$
In this case  $y_{ij}$ has a value equal to $\beta_{j}'$, otherwise $f_{ij}$
would be zero.
%If $y_{ij}=\beta_j$ then $\alpha_i = c_{ij}-p_{ij}$.
%Since $\alpha_i$ is greater than $0$, $p_{ij}\beta_j < c_{ij}$. 
Since the value of $\beta_j$  changes 
by a factor of $(1+\epsilon)$, $\alpha_i \leq c_{ij} - p_{ij}\beta_j + \epsilon c_{ij}$. 
%Note that source $i$ is chosen so that $f_{ij}$ is being pushed at the valuation
%$y_{ij} = \beta_j/(1 +\epsilon)$,  
Note that the value of $\beta_j$ can not
increase further unless $f_{ij}$ reduces to zero.
%For all other $k$, $\alpha_i \

If $\alpha_i$ has also changed, it only decreases, 
since $\beta_j$ is increasing. 
Therefore, the inequality still holds. 

\ignore{
In every iteration of the algorithm, after possible modification of
$\beta_j$, $\alpha_i$ is recomputed. Note that $\beta_j$ can
only increase by a factor of $(1+ \epsilon)$. 
Further note that $c_{ij}-p_{ij}\beta_j \geq 0$ during the course of the
algorithm. 
The first time that $c_{ij}-p_{ij}\beta_j \leq 0$,
$\alpha_i$ becomes zero 
and $y_{ij}$ is set to $\beta_j'$ so as to
ensure that flow is pushed back on edge $ij$. 
This would result in $f_{ij}$ possibly becoming zero.
Also, subsequently source $i$ is never considered  for pushing
flow $f_{ij}$ since line 1 of {\bf Algorithm 2.2 Auction} only
considers vertices with $\alpha_i >0$.
Thus $\beta_j$ would only increase if $f_{ij}=0$.
We consider cases
based on the value of $\alpha_i$ computed in line 24 of the algorithm,
{\bf Algorithm 2.2 Auction}:

\begin{enumerate}
}
\item
$\alpha_i =0$. No future change in the value of $\alpha_i$ can occur.
Prior to the step at which $\alpha_i$ is set to zero,
$\alpha_i = c_{ij}-p_{ij}\beta_j > 0 $ for some $j$.
This implies that $p_{ij}\beta_j < c_{ij}$. An increase in
the value of $\beta_j$ thus  still results in
$c_{ij}-p_{ij}\beta_j + \epsilon c_{ij} > 0$.
Thus the inequality holds when $\alpha_i$ is set to zero.
Further, when $\alpha_i$ becomes zero, 
the algorithm  sets the value of $y_{ij}$ to be $\beta_j'$ so as to
ensure that $\beta_j$ is not increased without the event that
flow is pushed back on  $ij$ and $f_{ij}$ is reduced to zero. 
And, subsequently, source $i$ is never considered  for pushing
flow $f_{ij}$ since line 1 of {\bf Algorithm 2.2 Auction} only
considers vertices with $\alpha_i >0$.
Thus if flow $f_{ij}>0$ then $\beta_j$ has not changed. And if
$\beta_j$ has changed, the flow is set to zero.
The inequality thus holds.

\end{enumerate}

\ignore{
We get

\[
\forall ij:f_{ij} > 0,  
\alpha_i \leq c_{ij} - p_{ij}\beta_j +  \epsilon c_{ij}
\] 
by combining all the cases above.
}%%%%IGNORE
\end{proof}

Furthermore,
\begin{lemma}
Algorithm 2.2 {\bf Auction}  terminates.
\end{lemma}
\begin{proof}
At each iteration of the algorithm, flow is pushed on an edge and 
$ \forall j, \beta_j$ either remains the same or
the value of $\beta$ increases for some $j$ (Lemma~\ref{beta-change}). 

When flow is pushed on an edge $ij$, either (i) $s_i$ goes to zero  (ii) the budget at a
sink $d_j$ is met or  (iii) the flow on a back edge $ji'$ is reduced to zero.
Note that, by lemma~\ref{beta-increase}, in case (iii) flow cannot be pushed (in the forward direction)
on the edge $i'j$ without
increasing the value of $y_{ij}$. 

The first event, case (i), results in flow being pushed from $i$ to $j$
such that $y_{ij}=\beta_j$ and will happen only once per source-sink pair
until an event where flow is pushed back on an edge $j'i$ to add to  surplus at $i$
occurs. Thus this event occurs whenever new surplus is generated, which
is bounded by the number of times  flow is pushed back on an edge without
case (iii) ocurring.
The number of times flow is pushed back on an edge, without reducing the
flow to zero (case (iii)), is bounded since 
the amount of change of flow at every instance is a rational with
bounded denominator.

The second event, case (ii),
happens  at most $O(m)$ times, since after that event
the sink remains saturated during the subsequent operations of the algorithm. 
The third event, case (iii), occurs at most $n$ times before all edges 
leading to a sink have flow reduced to zero, and consequently the value of
$\beta_j$ must rise.
This bounds the total number of operations
in the case that $\beta_j$ remains the same.
Further, the number of increases of $\beta_j$ is bounded since $\beta_j$ 
increases by a factor of $1+\epsilon$ and its value is bounded above
by $\max_j \{ c_{ij} \}$. Thus the algorithm terminates in a finite number of steps.
\end{proof}

%Note that the amount of flow to be pushed from each sink is bounded.
%And from properties of the linear program, 
%so is the value of $\beta_j , \forall j$.
%The algorithm converges, since, during the course of the algorithm,
%either new flow is being pushed or, for some $j$,
%the value of $\beta_j$ rises.
In order to show good convergence of the algorithm we 
introduce further modifications. However, when the algorithm terminates
then the following is true:
\vspace*{12pt}
\begin{lemma}
At termination, the difference between the value of 
primal solution $\sum_{ij} c_{ij}f_{ij}$ and the value of the 
dual solution is at most $ \epsilon \sum_{ij} c_{ij}f_{ij}$
\label{lemma_approx}
\end{lemma}

\begin{proof}
The value of the dual solution is 
\begin{eqnarray}
&=&\sum_i a_i\alpha_i  + \sum_j b_j\beta_j  \\
&=&\sum_i (a_i - \sum_j f_{ij})\alpha_i + \sum_j (b_j - \sum_i p_{ij}f_{ij})\beta_j + 
\sum_{ij}f_{ij}\alpha_i + \sum_{ij}f_{ij}p_{ij}\beta_j
\end{eqnarray}

When we subtract the value of the primal solution $\sum_{ij} c_{ij}f_{ij}$, from the
above, the  difference is
\begin{eqnarray}
\sum_i (a_i - \sum_j f_{ij})\alpha_i + \sum_j (b_j - \sum_i p_{ij}f_{ij})\beta_j
-\sum_{ij} f_{ij}(c_{ij} - \alpha_i - p_{ij}\beta_j)  
\end{eqnarray}

Thus, the total absolute difference is at most $\Delta_1 + \Delta_2 + \Delta_3$ where

\[
\Delta_1 = |\sum_{ij} f_{ij}(c_{ij} - \alpha_i - p_{ij}\beta_j)|
\]

\[
\Delta_2 = |\sum_{i} \alpha_i(a_{i} - \sum_j f_{ij})|
\]

\[
\Delta_3 = |\sum_{j} \beta_j(b_{j} - \sum_i p_{ij}f_{ij})|
\]

From lemma \ref{lemma_primalCS}, we have 

\[
\forall ij:f_{ij} > 0,  
\alpha_i \leq c_{ij} - p_{ij}\beta_j +  \epsilon c_{ij}
\]

Therefore, 

\[
\Delta_1 \leq  \epsilon \sum_{ij} c_{ij}f_{ij} 
\]

From the termination condition of the algorithm, we know that for any source
$i$ such that $a_{i} - \sum_j f_{ij} > 0$, we have $\alpha_i = 0$.

Therefore 
\[
\Delta_2 = 0
\]

%by setting $\delta = {\epsilon}{C\sum_i a_i}$ where $C = \max_{ij}c_{ij}$    

For all unsaturated sinks $\beta_j = 0$. Therefore, 
\[
\Delta_3 = 0
\]

Combining all of the above we get the lemma.
\end{proof}

\label{section_auction}

\subsection{Modified Algorithm}
\label{section_modified}
\label{sec_modified}

In order to show the required complexity bound, and make its analysis simpler, we make
a few modifications to the above algorithm. We address the details of
the data structures required and present the complexity analysis.
We also show that the modifications do not affect the correctness of 
algorithm \ref{alg_auc} {\bf Auction} as proved in section \ref{sec_auc}.

We introduce the concept of a preferred edge, back edge and derived graph 
as explained below. %These aid in the analysis of the complexity.

\paragraph*{Preferred Edge: }
%{\bf Preferred Edge: } 
For a given set of values for $\alpha_i$ and $\beta_j$ we designate
a preferred edge $Pr_i$ for each source $i$. This edge maximizes
$c_{ij}-p_{ij}\beta_j$, i.e. $Pr_i =  ik$ where 
$ k= {\rm argmax}_j (c_{ij}-p_{ij}\beta_j)$. 
In the case that multiple values of $j$ satisfy the condition, only one is picked.
The source $i$ will be required to push
flow along this edge until the edge is no longer is preferred. 

\paragraph*{Back-edge: }
%{\bf Back-edge: }
All edges $ij$ such that $f_{ij} > 0$ and $y_{ij}<\beta_j$.
Flow is pushed back on these edges.
The set of all such edges is termed $B$.

\paragraph*{Derived Graph: }
%{\bf Derived Graph: }
We define the capacitated derived graph $H$ with respect to an intermediate solution
maintained by the modified algorithm. This graph consists of all
the sources and sinks in the given bipartite graph and directed edges 
which are either the preferred edges or the back edges.
All the preferred edges are oriented in the forward direction, whereas 
all the back-edges are oriented in the reverse direction, indicating the
directions along which flow will be pushed. The capacity of a preferred edges
is infinite whereas the capacity of a back-edge is the amount of flow carried on
that particular edge. Formally,
$H=(S,T,E,c)$  where 
\[ E= \{ij | i \in S, j \in T, ij  \in Pr_i \} \cup
\{ ji , j \in T , i \in S,  ij \in B \} \] 
and the capacity function $c: E \rightarrow \mathbb{Z}^+$
is as follows:
\[ c(ij)= \infty  \ \ {\rm iff} \ ij \in Pr_i \]
\[ c(ji)= f_{ij}  \ \ {\rm iff} \ ij \in B \] 
%and there exists a capacity function:

%\paragraph{Phase}
%{\bf Phase: }
%The steps executed by the algorithm between two consecutive updates of 
%any value $\beta_j$. 
%{\bf Cleanup: }
Before starting with the main algorithm, 
the following  preprocessing step eases the description: 
\paragraph{Pre-processing step:}
Consider a preferred edge $ij \in Pr_i$. If edge $ij$ is also a back edge (thus
creating a 2-cycle ($i\rightarrow j \rightarrow i$) in $H$) and 
is not the only back edge from $j$ then 
the valuation  at which flow is shipped from source $i$ to sink $j$,
i.e. $y_{ij} = \beta_j/(1+\epsilon)$,  is increased to $\beta_j$.
Note that this modification is not done when the sink has only
one back edge  $ij$ since the existence of at least one valuation equal to
$\beta_j/(1+ \epsilon)$ does not allow a  raise  in the price of sink $j$. 
This ensures:
\begin{enumerate}
% without
%triggering a change in $\beta_j$.
% if the edge $ij$ is also the 
%preferred edge for source $i$. This is not done in case $ji$ is the only back edge 
%on sink $j$. 
\item
If there is a 2-cycle on sink $j$, sink $j$ has only
one back edge. 

\item
$\beta_j$ does not change. 
%does not raise the price of $j$, $\beta_j$ 
%since the valuation of
%$j$ by $i$ was $\beta_j/(1+ \epsilon)$ and since 
\end{enumerate}
In the derived graph $H$, if $ij \in Pr_i$ before the modification then $ij \in Pr_i$ 
after the modification also,
and source $i$ will choose to push flow along its preferred edge.

The algorithm in section~\ref{sec_basic} pushed flow from a source to a sink 
causing flow
to be pushed back to some other source and consequent independant processing of
that  other source was performed.
The modified algorithm in this section, starting from a source with 
a surplus and positive $\alpha_i$, finds a path along the edges 
of the derived graph $H$ and pushes  the flow according to the type of path discovered,
thus ensuring that the surplus flow is processed in one push of flow along
a flow path.
The path found by the algorithm ends at either
(i)  an unsaturated sink,(ii) a source with $\alpha_i = 0$ or (iii) a cycle.

%The algorithm {\em attempts} to conserve the flow as it pushes
%along the path. In this way, it is ensured that the surplus on any 
%intermediate source does not increase.

%As described above, the discovered path in the derived graph may end in cycles.

\begin{figure}%[h]
\centerline{\epsfysize=200pt\epsfbox{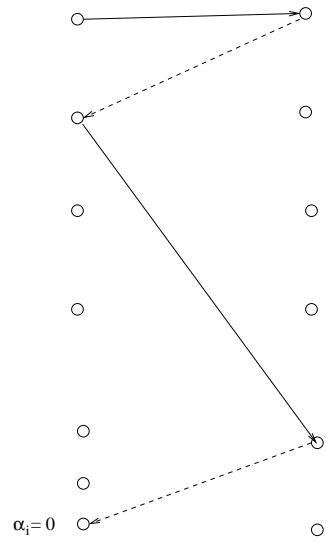}}
\caption{ An example where path ends in a source with $\alpha_i = 0$} 
\end{figure}

\begin{figure}%[h]
\centerline{\epsfysize=200pt\epsfbox{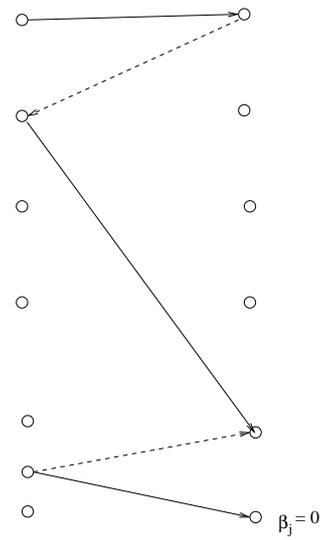}}
\caption{ An example where path ends in a sink with $\beta_i = 0$} 
\end{figure}

\begin{figure}%[h]
\centerline{\epsfysize=200pt\epsfbox{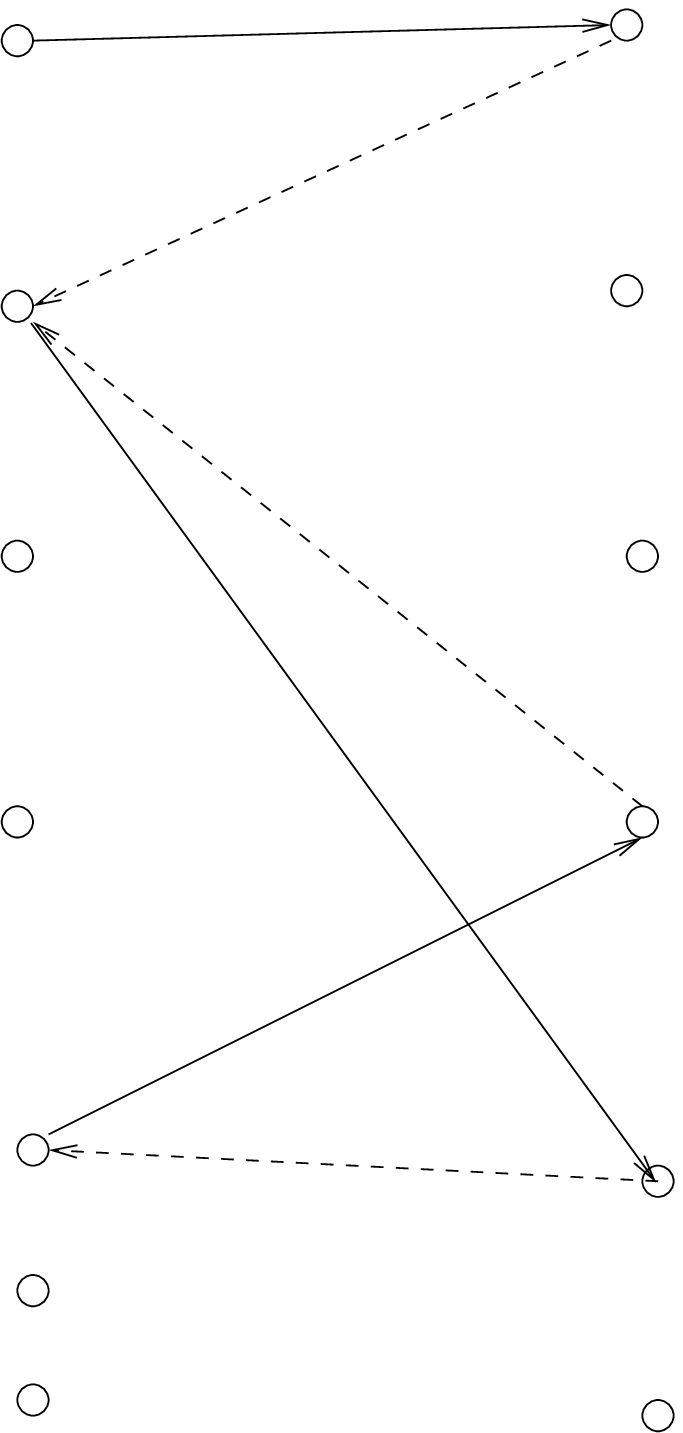}}
\caption{ Example where the path ends in a cycle}
\end{figure}

We detail the various cases that arise.
A source $i$ is picked such that $\alpha_i >0$ and 
$\sum_j f_{ij} < a_i $.
And  a  path, $P$, is discovered by, say a depth first search,
along the edges of the derived graph $H$.
There are  different types of paths/cycles that may arise.

\begin{enumerate}
\item
{\bf Type-1 paths:} We have two cases, one
where the path finds a  source $i_k$ where $\alpha_i = 0$ and the second where
the path finds  an unsaturated sink $j_k$ where $\beta_{j_k} = 0$. 
In the first case, let $P= i_0 , j_0 , \ldots i_k$
and in the second case let $P= i_0 , j_0 , \ldots , j_k$. 
In both cases flow is pushed as follows: in the first
case flow is pushed into 
source $i_k$ (i.e. returned to the source) and in the second
case flow is pushed into a sink ($j_k$), 
without violating feasibility or complementary slackness conditions. 
However, the flow may be limited by the capacity of the edges on the path, $P$.
We define the capacity of the path, $c(P)$ 
as $c(P) = \min  \{ c(e) | e \in  B \cap P   \}$
% \cup \{ d_{j_l} |  0 \leq l \leq k   $.
The processing of the flow is as follows:
%depends on the of flow that can be pushed is: 
\begin{enumerate}
\item
The flow is limited by the surplus, $s_{i_0}$, 
at the starting source $i_0$. Push all the surplus
along the path. The surplus at the source disappears.
This is true for both  cases.

\item
The amount of flow is limited by the capacity of an unsaturated sink on 
path $P$ (second case), i.e. $c(P) \geq d_{j_k}$. 
In this case, we push enough
flow to saturate the sink and consequently, raise the value of $\beta_j$. 

\item
The amount of flow is limited by the capacity of  the path $c(P)$.
In this case, we push flow  from the source along 
the path. Let $f$ be the flow starting at the source.
Let $e = j_l i_{l+1}$ be the first edge such that $f >  c(e)=c(j_l , i_{l+1})$.
The amount of flow that can be pushed along this edge is limited by $c(e)$.
Pushing flow along this edge reduces  the capacity of the edge  $e$ to $0$.
%If the flow on one of the back edges, directed from $j_l$ to $ i_{l+1}$ is reduced to zero,
The surplus $f-c(e)$  is added to the surplus of the previous source $i_l$ 
which has $i_l  j_l$ as the preferred
edge. 
%This surplus is $f- f_{i_{l_+1} j_l}$.
The flow pushed back through $j_l  i_{l+1}$ is accumulated at $i_{l+1}$ as a surplus.
The surplus at $i_{l+1}$  is pushed  further along the path recursively. 
In the end one or more back edges would have zero flow and surplus may be left at
each source ocurring just prior to the sink where the back edge with zero flow is incident.

\end{enumerate}

\item
{\bf Type-II paths:}
The path ends at a 2-cycle ($ji,ij$).
A  2-cycle implies $ji$ is the only back-edge on sink $j$, as soon as the flow
is again pushed along edge $ij$, $y_{ij}$ changes 
(line 15, {\bf Algorithm 2.2  Auction}) and the only back edge disappears. 
This results in an increase in the value of $\beta_j$.

\item
{\bf Type-III paths:}
A cycle is discovered since one of the sources is encountered again.
Let the path discovered be
denoted by $i_0,j_0...i_k...j_qi_k$ which is decomposed into
a simple cycle $C=i_k j_k \ldots j_q i_k$
and a simple path $P_1 = i_0 , j_0 \ldots i_k$.
%being the path from the with 
The source $i_k$ being the {\em entry point}, i.e.,
the first source that appears in the cycle. 
First the flow is pushed along the path $i_0...i_k$ 
taking into account the capacities as in Type-I paths.
Any additional surplus that appears at source $i_k$ is handled as follows. 

Consider an imaginary procedure that simulates sending flow around the cycle, 
$C $, multiple times, while conserving the
flow at each of the intermediate sources in the cycle.

%Consider the cycle $C= i_k...j_q i_{q+1}=i_k$.
We define  the {\em transfer ratio} $\rho_l$ for
a back-edge $j_{l}i_{l+1}$ w.r.t. $C$ 
to be  $p_{i_lj_l}/p_{i_{l+1}j_l}$ where $i_lj_l$ is the preferred edge for 
source $i$ and  $j_l i_{l+1}$ is a back edge.

Further define the {\em cumulative transfer ratio} $\rho^c_{l}$ 
for the back-edge $j_l,i_{l+1}$ with 
respect to a cycle $ C$
%\{i_k,j_k,i_{l+1},j_{l+1},...,j_q,i_k\}$ 
to be the amount that goes through the edge $j_l, i_{l+1}$
if one unit of flow is pushed along the cycle starting at source $i_k$. This is equal
to $\prod_{z = k}^{l} \rho_{z}$. 

Let the cumulative transfer ratio of the cycle be 
$\rho_\odot  = \prod_{l = k}^{q} \rho_{l}$. This is the amount of flow that reaches
the first source, $i_k$, back after being pushed all around the cycle starting with a unit of
flow at $i_k$. 

Let us call the push of flow once around the cycle as one revolution.
If we were to send the surplus at source $i_k$, $s_{i_k}$ units of flow,
repeatedly in this cycle for say, $r$ times, while conserving
the flow at each of the intermediate sources, the total flow shipped through the 
back edge $j_li_{l+1}$ is  given by the geometric 
series $F(l,r) = \sum_{z=0}^{r} s_{i_k} \rho^c_l (\rho_\odot)^z$. 
Since this back edge is capacitated, the maximum
number of times flow can be pushed around the cycle before saturating this edge to capacity,
which we term as the {\em limiting number of revolutions},  
is given by $R_l = {\rm argmax}_{r} 
(F(l,r) \leq f_{i_{l+1}j_l})$. 
This number can be computed in $O(1)$ time. 
It ignores that other edges may also limit the flow.
We thus compute $R_{min} = \min_{l \in B \cap C} R_l$ to compute
the {\em limiting number of revolutions for the cycle}, i.e.
the smallest value of the limiting number of revolutions, computed from  amongst
the back edges in the cycle, $C$.

The flow on the edges of the cycle is modified as:
\[f_{i_lj_l} \leftarrow f_{i_lj_l} +  
\sum_{j=0}^{R_{min}} s_{i_1} \rho^c_{l-1} (\rho_\odot)^{j}\]
for forward edges $i_lj_l$ and
%NEWSK
%\[$y_{i_zj_z} \leftarrow \beta_{j_z}$\]
\[f_{i_{l+1}j_l} \leftarrow f_{i_{l+1},j_l} - 
\sum_{j=0}^{R_{min}} s_{i_1} \rho^c_l (\rho_\odot)^{j}\]
corresponding to back edges.

We consider cases based on the value of $\rho_\odot$:
\begin{enumerate}

\item
$\rho_\odot < 1$. In this case we have a decreasing gain cycle, $C$,
i.e. the value of the flow
decreases as it is sent around the cycle. We compute the limiting number of 
revolutions for each back edge in the cycle, $C$. 
Note that this number may be infinity for some edges. 
However, this can be determined  
by computing the value to which the geometric series, $F(l,)$, converges 
and does not require simulating an infinite sequence of pushes. 
%Consider the source $i_k$ which is the entry point of $C$. Suppose the
%surplus pushed onto the the node  $i_k$ from the node $i_0$, 
%is $s_{i_k}$ We seek
%to diminish the surplus on this source while keeping the surplus at other nodes the same.

%Then we compute the limiting number of revolutions for $i_k$ . This is the number of 
%revolutions after which the surplus remaining at source ${i_k}$ is $\leq \delta$. This
%number is given by ${\rm argmin}_r (s_{i_k}\rho_{\odot}^r < \delta)$

%Let the smallest value of the limiting number of revolutions, computed from  amongst
%the back edges in the cycle, $C$,
%be $r_{min}$, i.e. $r_{min}= \min_{p \in B } \{R_p \}$.

%The flow on edge $(i_{l+1}j_l)$ is 
%$s_{i_k}F(l,R_{min})$ that is the value of the flow caried by the edge 
%as a result of the  simulation  of pushing of flow $R_{min}$ times. 
After $R_{min}$ rounds, the remaining
surplus at $i_k$ is reduced to $s_{i_k}\rho_ \odot ^ {R_{min}}$. 
Note that if $R_{min}$ is infinity then the surplus at the source, $s_{i_k}$ decreases to zero.
The surplus at every other
source in the cycle remains unchanged.

Consider the cases depending on the value of the surplus at $i_k$:

(i) the surplus remaining at $i_k$ becomes zero, or,
%less than $\delta$ 

(ii) the flow on one of the back edges reduces to zero.
This case occurs when the number of revolutions is limited due to one 
of the back-edges.
Since $R_{min}$ integral revolutions need not achieve this, one more sequence
of flow pushes 
%Note that this
%is achieved  by pushing the 
%flow {\em one more time } (after $R_{min}$  revolutions)
around the cycle may be required. 

\item
$\rho_\odot \geq 1$. In this case the flow increases or remains the same as it 
goes along the
cycle. There  will be thus no decrease of surplus at the starting source vertex $k$. 
However, for each of the back-edges 
$R_l = {\rm argmax}_{r} 
(F(l,r) \leq f_{i_{l+1}j_l})$
%$R_l = {\rm argmin}_{r} (F(l,r) \geq f_{i_{l+1}j_l})$ 
{\em is a finite number} since each push of flow around the cycle reduces the capacity of
the back-edges. 
The smallest of these  values is chosen and the maximum 
number of rounds of  flow 
to be pushed around the cycle is computed.
Again, since $R_{min}$ integral revolutions need not achieve this, 
it may  require one more sequence of flow pushes  around the cycle 
to reduce the flow on a back-edge to zero.

\end{enumerate}

In the cases above where the surplus is not reduced to zero,
at least one edge is saturated by the flow pushed around the cycle.
Any surplus remaining at other nodes is pushed around the cycle to accumulate
at the source(s) just before the saturated edge(s).

\end{enumerate}

The overall structure of the algorithm is similar to that in 
section~\ref{sec_auc}.
After pushing the flow through the path or cycle discovered, values of $\beta$
are updated if required, preferred and back edges computed as necessary and
the algorithm repeats until all surplus is removed or
until complementary slackness conditions are satisfied.
This algorithm is termed as the  {\em Modified Auction} algorithm.

\paragraph{Data Structures:}
In order to determine $\max_j c_{ij} -p_{ij}\beta_{j}$ for each $i$ we
maintain a heap. This heap can be updated when $\beta_j$ changes in
$O(\log{m})$ time. The UpdateBeta procedure, as in the previous algorithm,
utilizes a set data structure
to determine the need for update of $\beta_j$ in $O(1)$ time.

We next discuss how this approach  improves the complexity. However,
we can actually solve a more generalized problem, i.e. where 
each edge $(u,v), u \in S , v \in T$ has a capacity $u_{ij}$.
So we defer a formal description and  proofs of the above approach to
section~\ref{Edge-formal}.
%\subsection{Why does this work?}

%\paragraph{Preferred Edge}

%\paragraph{Cleanup}

\vspace*{12pt}
\subsection{A discussion on the correctness and complexity of Modified Auction } 

We discuss the correctness and  time complexity of the modified algorithm.
Formal proofs are provided when we consider the generalized model including
edge capacities.

\paragraph{Correctness:}
In order to justify that the modified algorithm 
terminates with a $(1-\epsilon)$ approximate solution, 
first, we observe that while the algorithm \ref{alg_auc} {\bf Auction} 
does not require any particular order 
in which the surplus sources are picked, the sequence of choices of 
the sources from where the surplus is reduced,  as imposed by 
the modifications, still ensures  primal and dual constraints.

In the case that the algorithm encounters cycles,
the modified algorithm 
essentially simulates multiple steps of the algorithm
\ref{alg_auc} {\bf Auction}. The end result is exactly the same if each step was 
performed individually
and repeatedly. The amount of flow shipped out of a source at each step, however, may be 
{\em smaller} than what Algorithm \ref{alg_auc} {\bf Auction} would have shipped out. 
However, the primal and dual constraints are still ensured.

The modified algorithm makes the same changes to the variables $\alpha_i$,
$\beta_j$ or $f_{ij}$ as described in the previous section. 
The only modification is in the change to the  variables, $y_{ij}$ in the 
preprocessing procedure. This 
change does not affect the claim that $\beta_j \leq y_{ij}(1+\epsilon)$ for  all $i$ in the proof
of lemma \ref{lemma_primalCS} as $y_{ij}$ is set to $\beta_j$ itself. 
The pre-processing procedure 
leaves at least one back-edge intact; this ensures that 
no change in $\beta_j$ takes place as 
a result of the preprocessing. All the arguments involving $\beta_j$ are, therefore, unaffected.
Thus primal and dual feasibility is maintained.
Further the complementary slackness conditions are satisfied.

\ignore{
Thus we can state that
\begin{itemize}

\item
The primal solution obtained by the algorithm {\em Modified Auction} is feasible.

\item
The dual solution obtained by the modified algorithm is feasible.

\end{itemize}
%does not affect the proofs as the amount shipped out still satisfies all the capacity constraints.  

Due to the fact that the change in $\beta_j$ and consequent changes in other variables are
locally ensured by the same procedure as  in the previous section, the following holds:
\begin{itemize}

\item
The complementary slackness conditions {\bf CS1} and {\bf CS2}
%\ref{source_cond} and \ref{sink_cond} 
are satisfied exactly and condition {\bf CS3}
%\ref{primal_CScond} 
is satisfied approximately.

\item
At termination the difference between the value of 
primal solution $\sum_{ij} c_{ij}f_{ij}$ and the value of the 
dual solution is at most $ \epsilon \sum_{ij} c_{ij}f_{ij}$

\end{itemize}
}%%%IGNOREEND
\paragraph{Complexity:}
For the  complexity
we consider the  sequence of steps between two 
successive raises of the value of $\beta$, termed as a {\bf phase}, in the algorithm.
%\paragraph{Data Structures}

We use the 
fact that for all $j$, the total number of times the value $\beta_j$ 
rises is bounded.
We show that each operation during the course of the algorithm
can be charged to a rise in the value of $\beta_j$ for some $j$.
In fact, $\forall \beta_j$, every rise in the value of $\beta_j$ is required 
to be charged at most $3n^2+n$  times during the  algorithm.
We account for the work required for each push of flow along a path or cycle
via a charging argument. 
There are at most $n$ pushes of flow along edges of a path
before either:
%\begin{enumerate}
%\item
(i) a 2-cycle  is encountered,
\ignore{
If the 2-cycle is between $i$ and $j$, source $i$ increases the
assigned value $y_{ij}$ to $\beta_j$. Since $ji$ was the only back-edge
on $j$, there is a change in $\beta_j$. We charge the $O(n)$ pushes 
to this rise in $\beta_j$.
}
(ii) a source with $\alpha=0$ is encountered, 
%This case is analyzed together with the next case.
(iii) an unsaturated sink is encountered,
or (iv) a cycle is encountered. 
Note that in the last case all the calculations for 
determining the limiting number of
revolutions can be done in $O(n)$ time, since the cumulative transfer ratio can be 
determined by traversing the cycle once. As noted above, in all the cases
that arise when flow is pushed around the cycle,
either the surplus that has been pushed into the cycle reduces to 0 %$\delta$, 
or the flow on one of the back edges is reduced to zero. The $O(n)$
amount of work can be charged to the disappearance 
%(to a negligible amount $\delta$) 
of surplus at a source or reduction of flow to zero on a back-edge $ji$,
and hence an
increase in  $\beta_j$, subsequently, when all back edges incident to
the sink $j$ have the corresponding flow through them reduced to zero.
The other cases have a similar analysis.
The detailed complexity of the algorithm will be described in the next section.

\section{Edge Capacities}
\label{Edge-formal}

In this section we further generalize the Budgeted Transportation Problem. In this version
each edge has an upper bound on the flow going through. The capacity of each edge
is represented by $u_{ij} \in \mathbb{Z}^+$. We call this problem the Budgeted
Transshipment (BTS) problem. It is easy to see that the budgeted transportation is a special
case where edge capacities are infinite. 

Consider the Linear Program for BTS 

\begin{eqnarray}
{\rm maximize\ \ \ \ \ }\sum_{ij} c_{ij}f_{ij}
\end{eqnarray}

\ \  subject to:
\begin{eqnarray}
\sum_{j} f_{ij} &\leq& a_i \ \ \ \forall i\\
\sum_{i} p_{ij}f_{ij} &\leq& b_j \ \ \ \forall j\\
f_{ij} &\leq& u_{ij} \ \ \ \forall ij\\
f_{ij} &\geq& 0 \ \ \ \forall ij
\end{eqnarray}

The dual of the above program is 

\begin{eqnarray}
{\rm minimize\ \ \ \ \ }\sum_{i} \alpha_i a_i +  \sum_j \beta_j b_j + \sum_{ij} u_{ij}\gamma_{ij}
\end{eqnarray}

\ \  subject to:
\begin{eqnarray}
\alpha_i \geq c_{ij} - p_{ij}\beta_j - \gamma_{ij}\ \ \ \ \forall ij\\
\alpha_i, \beta_j, \gamma_{ij} \geq 0 \ \ \ \ \ \ \ \ \ \ \ \  \  \forall ij 
\end{eqnarray}

And the complementary slackness conditions are

\begin{eqnarray}
f_{ij}(c_{ij} - \alpha_{i} - p_{ij}\beta_j - \gamma_{ij}) &=& 0 
\label{flow_cs_cond} \ \ \ \ \forall ij\\
\alpha_i (a_{i} - \sum_j f_{ij}) &=& 0 \ \ \ \ \forall i
\label{source_cs_cond}\\
\beta_j(b_j - \sum_i p_{ij}f_{ij}) &=& 0 \ \ \ \ \forall j
\label{sink_cs_cond}\\
\gamma_{ij}(u_{ij} - f_{ij}) &=& 0 \ \ \ \ \forall ij
\label{edge_cs_cond}
\end{eqnarray}

\section{Approximation Algorithm for BTS}
\label{sec_extend}

We now present a detailed  approximation algorithm for the BTS problem.

We first (re)define some terms

\paragraph*{Edge dual variable $\gamma_{ij}$: }
%{\bf Edge dual variable $\gamma_{ij}$: } 
The dual variable associated with each  edge. If $ij$ is saturated, 
$\gamma_{ij}$ is 
set as,  $\gamma_{ij} =\min_{j'} \{ (c_{ij} - p_{ij}\beta_{j})-(c_{ij'} - p_{ij'}\beta_{j'}) \}$ 
over all choices of edges  $ij'$.
It is set to 0  if edge $ij$ is unsaturated.

\paragraph*{Preferred Edge: } 
%{\bf Preferred Edge : } 
For each source $i$ an edge $ij$, termed $Pr_i $, such that $ij$ is 
unsaturated and $c_{ij}-p_{ij}\beta_j$ is maximized. Note that in the case 
when more than one edge meet the criteria, one of them is chosen as the preferred
edge. The preferred edges are directed forward (from $i$ to $j$) in the derived graph. 
Therefore, every forward edge is unsaturated. 
\paragraph*{Back Edge: }
%{\bf Back Edge : }
An edge $ij$ such that $f_{ij} > 0$ and $y_{ij} = \frac{\beta_j}{1+\epsilon}$.\\

As in the previous algorithm, a derived graph $H$ is used by the algorithm.
The derived graph has vertices as the original vertex set and edges as back edges
and preferred edges. The edges are directed
with the preferred edge from source $i$, $Pr_i= ij$, being directed from source
$i$ to a sink $j$, with capacity given by $u_{ij}-f_{ij}$, 
and a back edge directed from a sink $j$ to a source $i'$ with capacity $f_{i'j}$.

The algorithm differs from the one in section \ref{sec_auc} in  the
use of the variables $\gamma_{ij}$. 
To ensure complementary slackness,
this variable becomes non-zero only if 
edge $ij$ is saturated. 
This variable is not explicitly maintained but its value can be deduced from
$\alpha_i$ and $\beta_j$.
Its value indicates the extra cost incurred in
using another edge from source $i$.
To ensure complementary slackness conditions,
the value of the variable $\gamma_{ij}$ implicitly reduces when $\beta_j$
rises and when 
the reduced value becomes zero, the edge $ij$ can be designated as a back-edge
such that flow can be pushed back on 
the edge $ij$, leaving it unsaturated.

The following additions/modifications are made to the algorithm:

\begin{itemize}

\item
During the initialization, $\gamma_{ij}$ is set to zero.

\item
$\beta_j$ update: This update occurs when  $y_{ij}=\beta_j$,
for all unsaturated edges with positive flow. 
See procedure \ref{proc_update}
{\bf update}.

\item
Preprocessing : Before the removal of 2-cycles, for each saturated edge
$ij$, if $\gamma_{ij} > 0 $, $ij$ cannot be  a back edge. In case 
$\gamma_{ij} \leq  0$, it becomes a back edge. (See procedure \ref{proc_update} 
{\bf update}).

\item
Path/Cycle discovery  and pushing 
of flow: In the BTS problem, the forward edges have a capacity 
as well. While pushing the flow in a path or in a cycle, the bottleneck edges 
could be forward edges in addition to back-edges. 
(See \ref{proc_findpath} {\bf findpath}, \ref{proc_pushpath} {\bf
pushFlowPath} and \ref{proc_pushcycle} {\bf pushFlowCycle}).

\end{itemize}

The algorithm with the above modification is described as Algorithm \ref{alg_modauc} 
{\bf Modified Auction},
along with the associated procedures. 

The algorithm initializes the
primal and dual solution using procedure \ref{proc_modinit} {\bf Initialize}. 
Preprocessing is done in procedure \ref{proc_preprocess} {\bf Preprocess}
to choose the 
preferred edge and eliminate 2-cycles.

It then repeats the following steps:  find a path in the derived 
graph using \ref{proc_findpath} {\bf findpath}; depending on whether the path found
is a simple 
path or contains a cycle, the algorithm pushes flow  using \ref{proc_pushpath} {\bf pushFlowPath} or 
\ref{proc_pushcycle} {\bf pushFlowCycle}. 

Procedure \ref{proc_pushpath} {\bf pushFlowPath} starts with the first source and
moves flow along the path without accumulating it at any of the intermediate sources, unless,
one of the edges on the path is saturated. In this case, the source preceding the edge
accumulates some surplus.

Procedure \ref{proc_pushcycle} {\bf pushFlowCycle} is more involved. It uses 
Procedure \ref{proc_pushpath} {\bf pushFlowPath} to transfer the surplus from the initial source
to the first source that the algorithm encounters in the cyclic part of the path 
(See figure \ref{fig_cyclepush2}).
It then computes the number of pushes of flow through the cycles that are required
to (i) identify bottleneck edges or (ii) push the surplus to zero
% simulates multiple pushing along the cycle by computing appropriate flows 
(See figure \ref{fig_cyclepush3}).
% and \ref{fig_cyclepush3}). 
Note that an extra 
last push may be required 
to transfer the flow from the first source in the cycle to the one just before the 
bottleneck edge as illustrated in figure \ref{fig_cyclepush4}.

When required, $\beta_j$ is updated using Procedure \ref{proc_update} {\bf update} and 
preprocessing is done for the next iteration. It also checks if a  saturated edge
$ij$  which has a small value of $\gamma_{ij}$  prior to the
increase in $\beta_j$ has the condition
that $ \gamma_{ij} = c_{ij} - p_{ij}\beta_j - \alpha_i <0$,
i.e. the edge is on the verge (i.e. $\gamma_{ij}$ is close to $0$) 
of becoming unsaturated and 
allows for the edge to become unsaturated by designating the edge
as a back-edge. 

Each source maintains the effective profit $c_{ij} -p_{ij}\beta_j$ for all the 
unsaturated edges in a heap. Whenever an edge is saturated, it is removed 
from the heap. When $\gamma_{ij}$ is set to 0, it is re-inserted into the
heap. Each of these operations requires $O(\log{m})$ time. These operations are performed
when $\beta_j$ changes, for any $j$. 

The algorithms terminates when all the complementary slackness conditions are approximately 
met. The algorithm can be analyzed for correctness as well as complexity using the same 
framework as section \ref{section_modified}.

The correctness is shown by proving that both primal and dual solutions are feasible throughout
the execution and complementary slackness conditions are approximately met at termination.

The complexity is analyzed via a charging argument. Each iteration of the algorithm is
charged to one of the events: clearing of surplus at a source, saturation of a sink or 
saturation of an edge. All such charges are then accounted for in the
final analysis.

The details of the algorithm and the proofs are presented in the following section.

%\Subsection{Data Structures}

\subsection{Detailed Algorithm}
\label{sec_details}

 The modified  algorithm for BTS is presented below in detail, followed by 
proof of correctness and complexity analysis. Since BTP is a special case
of BTS, the algorithm is applicable to both the problems.

\begin{algorithm}
\caption{ {\bf Modified Auction}}
\begin{algorithmic}[1]

\vspace*{.1in}
\STATE{ {\bf Initialize}}
\STATE{{\bf Preprocess}}
\WHILE{ there exists a source $i$ such that $\sum_j f_{ij} < a_i $ {\bf and} $\alpha_i > 0$}
\label{alg_modauc:line1}
\STATE{ $P \leftarrow$ {\bf FindPath}($i$)}
\IF{$P$ contains a cycle $C$ which is not a 2-cycle at the end of $P$}
\STATE{ {\bf pushFlowPath} ($P'$) where $P'$ is the portion of $P$ before $C$}
\STATE{ {\bf pushFlowCycle }($C$)}
\ELSE
\IF{$P$ does not end with a 2-cycle}
\STATE{ {\bf pushFlowPath}($P$)}
\ELSE
\STATE{//$P$ ends at a 2-cycle $i_{l}j_l$ }
\STATE{ {\bf pushFlowPath}($P'$) where $P'$ excludes the last edge $i_lj_l$}
\STATE{ Eliminate back edge $j_li_l$ and let
$y_{i_lj_l} \leftarrow \beta_{j_l}$}

\ENDIF
\ENDIF
\STATE{{\bf Update $\beta_j$}}
\ENDWHILE

\end{algorithmic}
\label{alg_modauc}
\end{algorithm}

\begin{algorithm}
\caption{{\bf Initialize}}
\begin{algorithmic}[1]
\vspace*{.1in}

\STATE{$f_{ij} \leftarrow 0$ for all $ij$}
\label{proc_modinit:line1}
\STATE{$\beta_j \leftarrow 0$ for all $j$}
\label{proc_modinit:line2}

\STATE{$\gamma_{ij} \leftarrow 0$ for all $ij$}
\label{proc_modinit:line3}
\STATE{$\alpha_i \leftarrow \max_j c_{ij}$ for all $i$}

\end{algorithmic} 

\label{proc_modinit}
\end{algorithm}

\begin{algorithm}
\caption{\bf Preprocess} 
\begin{algorithmic}[1]
\vspace*{.1in}

\FORALL {$i$ such that $\sum_j f_{ij} < a_i$ {\bf and } $\alpha_i > 0$}
\STATE{Find an unsaturated $ij$ such that $c_{ij} - p_{ij}\beta_j$ is maximized}
\STATE{Make $ij$ the preferred edge for source $i$}
\STATE{$\alpha_i \leftarrow \max(0,c_{ij} - p_{ij}\beta_j)$} 
\label{proc_preprocess:line1}
\STATE{//Remove all 2-cycles unless it involves the only back edge}
\IF{$\exists i'$ such that $y_{i'j} < \beta_j$}
\STATE {$y_{ij} \leftarrow \beta_j$} 
\ENDIF
\ENDFOR

\end{algorithmic}
\label{proc_preprocess}
\end{algorithm}

\begin{algorithm}
\caption{{\bf $\beta_j$ update} } 
\begin{algorithmic}[1]
\vspace*{.1in}

\FORALL {$j$}
\STATE{ //Update $\beta_j$ if all assigned sources are shipping at the higher value}
\IF{ $\forall i$ such that edge $ij$ is unsaturated and  $y_{ij} = \beta_j$} 
\STATE {$\beta_j \leftarrow \beta_j(1+\epsilon)$ }
\label{proc_update:line1}
\FORALL {$i$ such that edge $ij$ is saturated}
\IF{$c_{ij} - p_{ij}\beta_j - \alpha_i < 0$}
\label{proc_update:line3}
%\STATE{//There is an unsaturated edge with a larger effective profit}
\STATE{//There is an edge with $\gamma_{ij} $ close to zero, i.e. there is another edge with a larger effective profit}
\STATE{Designate $ij$ as a back-edge}
\label{proc_update:line2}
\STATE{//Raise $\beta_j$ to $\epsilon\min_i \frac{c_{ij}}{p_{ij}}$ if it is 0 and sink 
$j$ is saturated}
\IF{$\sum_i f_{ij} p_{ij} = b_j$ {\bf and} $\beta_j = 0$}
\STATE{$\beta_j \leftarrow \epsilon\min_i \frac{c_{ij}}{p_{ij}}$}
\label{proc_update:line10}
\ENDIF

\ENDIF 
\ENDFOR
\ENDIF
\STATE{{\bf Preprocess}}
\ENDFOR

\end{algorithmic}
\label{proc_update}
\end{algorithm}

%\begin{algorithm}
%\caption{{\bf pushFlowPath}($P$)}
%\begin{algorithmic}[1]
%\ENSURE{$P = \{i_1,j_1,i_2,j_2...i_l,j_l\}$}

%\FOR{$k \leftarrow 1$ to $l$} 
%\STATE{ Find the limiting amount of flow for each forward edge $\phi_{i_kj_k}$ 
%        and each back edge $\phi_{i_{k+1},j_k}$ assuming the starting flow is unit}

%\STATE{flow Multiplier$\mu_{i_kj_k} \leftarrow \prod_{z=1}^{k-1} \frac{p_{i_zj_z}}{p_{i_{z+1}j_z}}$}
%\STATE{flow Multiplier $\mu_{i_{k+1}j_k} \leftarrow \prod_{z=1}^{k} \frac{p_{i_zj_z}}{p_{i_{z+1}j_z}}$}
%\STATE{available capacity $\xi_{i_kj_k} \leftarrow u_{i_kj_k} - f_{i_kj_k}$}
%\STATE{available capacity $\xi_{i_{k+1}j_k} \leftarrow f_{i_{k+1}j_k}$}
%\STATE{limiting amount of flow $\phi_{i_kj_k} = \frac{\xi_{i_kj_k}}{\mu_{i_kj_k}}$}
%\STATE{limiting amount of flow $\phi_{i_kj_k} = \frac{\xi_{i_{k+1}j_k}}{\mu_{i_{k+1}j_k}}$}

%\ENDFOR

%\STATE{Find the limiting edge $(i',j')$for which $\phi_{i'j'}$ is minimized}
%\IF{$\phi_{i'j'} < s_i$}
%\STATE{Push $\phi_{i'j'}$ along the path $P$}
%\ELSE
%\STATE{Push $\phi_{i'j'}$ along the path $P$}
%\ENDIF

%\STATE{return $(i'j')$}

%\end{algorithmic}
%\end{algorithm}

\begin{algorithm}
\caption{{\bf pushFlowPath}($P$)}
\begin{algorithmic}[1]
\vspace*{.1in}
\ENSURE{$P = \{i_1,j_1,i_2,j_2...i_l,j_l\}$ or $P= \{ i_1 , j_1 \ldots i_l , j_l , i_{l+1} \}$}

\STATE{ //$\phi$ is the amount that is transferred. Initially, it is the surplus at the first source
in the path}
\STATE{$\phi \leftarrow s_i$}
\STATE{ $r=l$}
\IF{$P$ ends at $j_{l}$}
\STATE $r= l-1$
\ENDIF
\FOR{$k \leftarrow 1$ to $r$} 
%\IF {$\phi > u_ij - f_{ij}$}
\STATE//The transferable amount is subject to the forward and backward edge capacity constraints
\STATE{$\phi \leftarrow \min(\phi,u_{i_kj_k}-f_{i_kj_k},f_{i_{k+1}j_{k}}\frac{p_{i_{k+1}j_{k}}}{p_{i_{k}j_{k}}})$}
\label{proc_pushpath:line1}
%\STATE{$\phi \leftarrow u_{ij} - f_{ij}$}
%\STATE{$f_{ij} \leftarrow f_{ij} + \phi $}
%\STATE{$e \leftarrow i_k,j_k$}
%\STATE{return $e$}
%\ELSIF {$\phi_i \frac{p_{i_{k}j_{k}}}{p_{i_{k+1}j_{k+1}}} < f_{i_{k+1}j_{k+1}} $}
%\STATE {$e \leftarrow i_{k+1}j_{k+1}$}
%\STATE {$f_{ij} \leftarrow f_{ij} +  f_{i'j}\frac{p_{i_{k+1}j_{k+1}}}{p_{i_{k}j_{k}}}$}
%\STATE {$\phi \leftarrow f_{i'j}$}
%\STATE {$f_{i'j} \leftarrow 0$}
%\STATE{return $e$}
%\ELSE
\STATE {$f_{i_kj_k} \leftarrow f_{i_kj_k} + \phi$}
\label{proc_pushpath:line4}
%NEWSK
\STATE {$y_{i_kj_k} \leftarrow \beta_{j_k}$}
\STATE{//$\phi$ is multiplied by the fraction of two prices to reflect the proportional\\
//flow reduction on back edge $i_{k+1}j_k$}
\STATE {$\phi_\leftarrow \phi \frac{p_{i_{k}j_{k}}}{p_{i_{k+1}j_{k}}} $}
\label{proc_pushpath:line2}
\STATE {$f_{i_{k+1}j_{k}} \leftarrow f_{i_{k+1}j_{k}} - \phi$}
\ENDFOR
\IF{$r=l-1$}
\STATE{$\phi \leftarrow \min(\phi,u_{i_lj_l}-f_{i_lj_l}$ })
%,f_{i_{k+1}j_{k}}\frac{p_{i_{k+1}j_{k}}}{p_{i_{k}j_{k}}})$}
\STATE {$f_{i_lj_l} \leftarrow f_{i_lj_l} + \phi$}
%\STATE {$y_{i_lj_l} \leftarrow \beta_{j_l}$}
\ENDIF
\label{proc_pushpath:line3}
%\ENDIF

%\STATE{return e}

\end{algorithmic}
\label{proc_pushpath}
\end{algorithm}

\begin{algorithm}
\caption{{\bf pushFlowCycle}($C$)}
\begin{algorithmic}[1]
\vspace*{.1in}

\REQUIRE{$C = \{i_1,j_1,i_2,j_2,...,i_k,j_k,i_{k+1}=i_1\}$ is a cycle in the derived graph}

\STATE{//Compute $\rho_\odot$}
%\STATE{$\rho_\odot \leftarrow pk_{i_1,j_k}\prod_{z=1}^{k-1} \frac{p_{i_zj_z}}{p_{i_{z+1}j_{z}}}$ }
\STATE{$\rho_\odot \leftarrow p_{i_kj_k}/p_{i_1j_k}\prod_{z=1}^{k-1} \frac{p_{i_zj_z}}{p_{i_{z+1}j_{z}}}$ }
\STATE{//Compute limiting number of revolutions for each edge}
\FOR{$l \leftarrow 1,k$}
\IF {$l>1$}
\STATE{$\rho^{'c}_l \leftarrow \prod_{z=1}^{l-1} \frac{p_{i_zj_z}}{p_{i_{z+1}j_{z}}}$}
\STATE{$\rho^c_l \leftarrow \prod_{z=1}^{l} \frac{p_{i_zj_z}}{p_{i_{z+1}j_{z}}}$}
\ELSE
\STATE{$\rho^{'c}_1 \leftarrow 1$}
\STATE{$\rho^c_1 \leftarrow  p_{i_1j_1}/p_{i_2j_1}$}
\ENDIF
\STATE{$ R'_l \leftarrow {\rm argmax}_r (\sum_{z=0}^{r} s_{i_1} \rho^{'c}_l (\rho_\odot)^z \leq u_{ij} - f_{i_lj_l})$}
\STATE{$ R_l \leftarrow {\rm argmax}_r (\sum_{z=0}^{r} s_{i_1} \rho^c_l (\rho_\odot)^z \leq  f_{i_{l+1}j})$}
\ENDFOR
\STATE{$R_{min} \leftarrow \min_l (R_l,R'_l)$}   
\STATE{// $R_{min}$ can be $\infty$} 
\STATE{$l^1_{min} \leftarrow {\rm argmin}_l (R'_l)$}
\STATE{$l^2_{min} \leftarrow {\rm argmin}_l (R_l)$}
\STATE{$l_{min} \leftarrow \min(l^1_{min},l^2_{min})$}

\vspace*{.1in}
\STATE{//change the flow on each edge to simulate pushing of multiple revolutions} 
\FOR{$z \leftarrow 1$ to $k$}
\IF{$R_{min} \neq \infty$}
\STATE{$f_{i_zj_z} \leftarrow f_{i_zj_z} +  
\sum_{j=0}^{R_{min}} s_{i_1} \rho^{'c}_{z} (\rho_\odot)^{j}$}
\ELSE
\STATE{//Compute the convergence of the infinite series }
\STATE{$f_{i_zj_z} \leftarrow f_{i_zj_z} +  
\sum_{j=0}^{\infty} s_{i_1} \rho^{'c}_{z} (\rho_\odot)^{j}$}
\ENDIF
%NEWSK
\STATE{$y_{i_zj_z} \leftarrow \beta_{j_z}$}
\STATE{$f_{i_{z+1}j_z} \leftarrow f_{i_{z+1},j_z} - 
\sum_{j=0}^{R_{min}} s_{i_1} \rho^{c}z (\rho_\odot)^{j}$}
\ENDFOR

\vspace*{.1in}
\IF{ $R_{min} $ is finite}
\STATE{//Send one more cycle of flow ; this will saturate some edge in one revolution}
\STATE{{\bf pushFlowPath}($i_1 \ldots i_k j_k i_1$)}

\vspace*{.1in}
\STATE{Let $P \leftarrow \{i_1,j_1,...i_{l-1},j_{l-1}, i_l \}$ where the edge $i_{l}j_l$ was 
saturated in previous step}
\STATE{{\bf pushFlowPath}($P$)}
\ENDIF

\end{algorithmic}
\label{proc_pushcycle}
\end{algorithm}

\begin{algorithm}
\caption{{\bf findPath}($i$)}
\begin{algorithmic}[1]
\vspace*{.1in}

\REPEAT
\IF {$\alpha_i > 0$ {\bf or} source $i$ already belongs to $P$} 
\STATE {return $P$}
\ENDIF
\STATE{Add source $i$ to $P$}
\STATE{Let $j$ be the sink such that $ij$ is $Preferred(i)$}
\label{line_xyz}
\IF {$\beta_j = 0$ {\bf or} sink $j$ already belongs to $P$}
\STATE { return $P$}
\ENDIF
\STATE{Add sink $j$ to $P$}
\STATE{Let $i'$ be a source such that $i'j$ is a back edge}
\STATE{$i \leftarrow i'$} 
\UNTIL{\bf false}

\end{algorithmic}
\label{proc_findpath}
\end{algorithm}

%\clearpage

%\begin{figure}[h]
%\centerline{\epsfysize=190pt\epsfbox{cyclePush1.eps}}
%\caption{Given a path ending in a cycle}
%\label{fig_cyclepush1}
%\end{figure}

\begin{figure}[h]
\centerline{\epsfysize=190pt\epsfbox{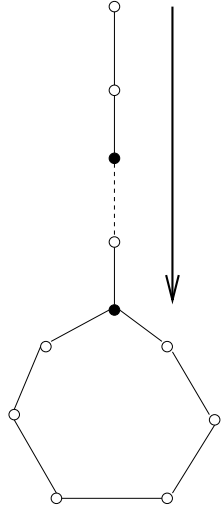}}
\caption{ First we bring the surplus to the first source in the cycle, possibly
          saturating an edge en route}
\label{fig_cyclepush2}
\end{figure}
%\clearpage

\begin{figure}[h]
\centerline{\epsfysize=190pt\epsfbox{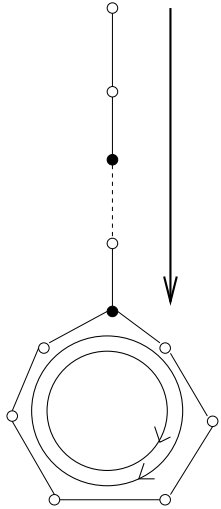}}
\caption{Send the flow around the cycle as many times as possible without
saturating any edge}
\label{fig_cyclepush3}
\end{figure}

\begin{figure}[h]
\centerline{\epsfysize=190pt\epsfbox{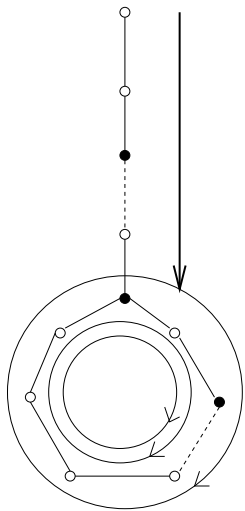}}
\caption{Send it one more time possibly saturating an edge and accumulating
surplus at the source before it}
\label{fig_cyclepush4}
\end{figure}

%\clearpage

%\begin{figure}[h]
%\centerline{\epsfysize=190pt\epsfbox{cyclePush5.eps}}
%\caption{Bring any surplus remaining at the first source to the one before
%the saturated edge}
%\end{figure}

\subsection{Feasibility of the Solutions}

\begin{lemma}
The primal and dual solutions maintained by algorithm \ref{alg_modauc} {\bf Modified Auction}
are feasible throughout its execution.
\label{lemma_BTSfeasible}
\end{lemma}

\begin{proof}
Consider the initial primal solution $f_{ij} = 0$ for all $f_{ij}$ (line \ref{proc_modinit:line1}
in Procedure \ref{proc_modinit} {\bf Initialize}). This is clearly a feasible primal solution.

Thereafter, the flow is increased or decreased in Procedures \ref{proc_pushpath} 
{\bf pushFlowPath}  and
\ref{proc_pushcycle} {\bf pushFlowCycle}. Let us consider each possible change.

\begin{enumerate}
\item
Increase/decrease of flow along a forward edge in line \ref{proc_pushpath:line1} of
Procedure \ref{proc_pushpath} {\bf pushFlowPath}. 

\begin{enumerate}
\item
The source constraint is satisfied as the increase in flow on edge $i_kj_k$ is at 
most the surplus generated at source $i_k$ in the previous iteration.

\item
Also, the increase is at most the available edge capacity $u_{i_kj_k}$, therefore
edge capacity constraint is satisfied.

\item
The total increase in the price of incoming flow on sink $j$ is $p_{i_kj_k}\phi$
(line \ref{proc_pushpath:line4}) and the total decrease is 
$p_{i_{k+1}{j}}\frac{p_{i_{k}j_{k}}}{p_{i_{k+1}j}}\phi$. Therefore, there
is no net change and the budget constraint remains satisfied and tight. 

\item
The decrease on edge $i_kj_k$ is $\phi \frac{p_{i_{k}j_{k}}}{p_{i_{k+1}j_{k}}}$ 
(line \ref{proc_pushpath:line2}, \ref{proc_pushpath:line3}) where $\phi$ is at most
$f_{i_{k+1}j_{k}}\frac{p_{i_{k+1}j_{k}}}{p_{i_{k}j_{k}}}$. Therefore the decrease
is at most $f_{i_{k+1}j_{k}}$. Thus the positivity constraint on edge flows
is satisfied.
\end{enumerate}

\item
Increase/Decrease in flow in Procedure \ref{proc_pushcycle} {\bf pushFlowCycle}. Everytime the
flow is pushed, the outgoing flow equals the incoming flow on every source
and a proportional amount out of every sink. Therefore,
not only are the source constraints satisfied, but also the
budget constraints on the sink.  The number of times the flow is pushed in the 
cycle is the minimum of  the limiting number of revolutions. Therefore the capacity 
constraints are satisfied. 

The last step uses Procedure \ref{proc_pushpath} {\bf pushFlowPath}, which satisfies the
constraints as explained above.

\end{enumerate}

Now consider the dual solution and its feasibility. For all $i$, $\alpha_i$ is 
initialized to $\max_j c_{ij}$ which is clearly larger than 
$\max c_{ij}-p_{ij}\beta_j-\gamma_{ij}$ as $\beta_j$ and $\gamma_{ij}$ are
zero to begin with (lines \ref{proc_modinit:line2} and \ref{proc_modinit:line3} in
Procedure \ref{proc_modinit} {\bf Initialize}). 

All the instances in the algorithm where these variables change are enumerated below. 
Assuming that the solution is feasible till these changes are made, we 
show the solution is still feasible after the change. 

\begin{enumerate}

\item
The update of $\beta_j$ in Procedure \ref{proc_update} $\beta_j$ {\bf Update}
at lines \ref{proc_update:line1} and \ref{proc_update:line10} increases it's value. 
Since any increase in  $\beta_j$ implies a decrease in the value 
$c_{ij} - p_{ij}\beta_j - \gamma_{ij}$, $\alpha_i$ remains larger 
than $c_{ij} - p_{ij}\beta_j - \gamma_{ij}$.

\item
Change in $\gamma_{ij}$ : $\gamma_{ij}$ is assigned a value 
equal to $\max (0,\alpha_{ij} - c_{ij'}- p_{ij}\beta_j')$. Therefore,
for every $ij$, $\alpha_{ij} \geq c_{ij} - p_{ij}\beta_j - \gamma_{ij}$ 
by construction.
%\rem{explain $\max( .. )$}

\item
Note that any change in $\alpha_i$ only occurs when there is a change
in some $\beta_j$, it does not change independently.  When it changes, it is set 
to $\max_{ij} (0,c_{ij} - p_{ij}\beta_j)$ (see line 
\ref{proc_preprocess:line1} in algorithm \ref{proc_preprocess} {\bf preprocess}). It is
therefore, larger than $c_{ij} - p_{ij}\beta_j - \gamma_{ij}$ for all $ij$ as 
$\gamma_{ij} \geq 0$.

\end{enumerate}
\end{proof}

\subsection{Complementary Slackness} We prove the following lemma about
the satisfaction of complementary slackness conditions.

\begin{lemma}
The algorithm terminates with the following conditions satisfied
\begin{eqnarray}
\forall \alpha_i > 0, \ \ \ a_i - \sum_j f_{ij} & = & 0
\label{source_cs_condedge}\\
\forall \beta_j > 0, \ \ \ b_j - \sum_i p_{ij}f_{ij} & = & 0
\label{sink_cs_condedge}\\
\forall \gamma_{ij} > 0,  \ \ \ u_{ij} - f_{ij}  & = & 0
\label{edge_cs_condedge}\\
\forall f_{ij} > 0 , \ \ \ \ |c_{ij} - \alpha_i - p_{ij}\beta_j - \gamma_{ij}| & \leq & \epsilon c_{ij} 
\label{approx_cs_condedge}
\end{eqnarray}
\label{lemma_BTSapprox}
\end{lemma}

\begin{proof} We consider each condition  

\begin{enumerate}

\item{Source slackness condition (\ref{source_cs_condedge}).}
These are the terminating 
condition for algorithm \ref{alg_modauc} (line \ref{alg_modauc:line1}). They are, 
therefore, satisfied at the end. 

\item{Sink slackness condition (\ref{sink_cs_condedge}).} 
Note that $\beta_j$ is
initialized to zero (Procedure \ref{proc_modinit}, line \ref{proc_modinit:line2}). Thereafter, $\beta_j$ is changed
only when sink $j$ is saturated (Procedure \ref{proc_update}, line \ref{proc_update:line10}). 
Also, once a sink is saturated, it stays saturated. This can be seen by observing that any
decrease of flow along any back-edge only takes place as a result of an increase in
proportional amount of flow along some other edge on this sink.

\item{Edge slackness condition (\ref{edge_cs_condedge}).} 
These are satisfied initially as $\gamma_{ij} = 0$. Thereafter,
as long as $\gamma_{ij}$ is greater than zero, the flow does not decrease on this edge as
it is never designated a back edge. Only when $\gamma_{ij}$ is set to zero 
(Procedure \ref{proc_update}, line \ref{proc_update:line3}),
%\ref{proc_preprocess:line1}), 
is the edge allowed to become a back edge. 

\item{Flow slackness condition (\ref{approx_cs_condedge}).} 
Similar to the proof of  lemma~\ref{lemma_primalCS}, we show that the following is  satisfied.
\begin{eqnarray}
\forall f_{ij} > 0, \ \ \alpha_i \leq c_{ij} - p_{ij}\beta_j - \gamma_{ij} + \epsilon c_{ij}
\label{app_cond_mod}
\end{eqnarray}

Consider unsaturated edges. Whenever a flow is increased along an edge
$ij$ it is a preferred edge. In Procedure \ref{proc_preprocess}, line \ref{proc_preprocess:line1},
whenever an edge is chosen as a preferred edge, $\alpha_i$ is set to $c_{ij}-p_{ij}\beta_j$ with
$\gamma_{ij}$ being zero as it is an unsaturated edge. The slackness condition is therefore, satisfied.
Subsequently, $\beta_j$ can rise by a factor of $(1+\epsilon)$ resulting in a rise of at most 
$\epsilon c_{ij}$ in effective profit, still satisfying (\ref{app_cond_mod}). Once 
$\beta_j$ has changed, $ij$ now becomes a back edge. Now, $\beta_j$ can not change any further
 unless this back edge has zero flow
(Procedure \ref{proc_update}, lines 3,4 and 8).
%\ref{proc_update:line3}). 
If as a consequence of a rise in $\beta_j$, $\alpha_i$ is set to zero, 
the inequality (\ref{app_cond_mod})
is true and  no further flow will be pushed from source $i$. 
$\beta_j$ cannot rise until all flow is pushed back to source $i$.

Next consider  when the edge is saturated.
In this case $\gamma_{ij}$ is set to $c_{ij} - p_{ij}\beta_j - \alpha_i$, thereby 
satisfying (\ref{app_cond_mod}).
However, flow may be pushed back on the saturated edge such that it
becomes unsaturated and $\gamma_{ij}$ is 
set to zero. This happens when $\beta_j$ rises. The change in the quantity 
$c_{ij}-p_{ij}\beta_j$ is no more than $\epsilon c_{ij}$; thus
 $|\gamma_{ij}| \leq \epsilon c_{ij}$ just prior to being set to zero. Thus
condition (\ref{app_cond_mod}) is satisfied after the change.
Condition (\ref{approx_cs_condedge}) immediately follows from condition (\ref{app_cond_mod}).  

\end{enumerate}
Hence we conclude the above lemma.
\end{proof}

From the above two lemmas 
we conclude:
\begin{theorem}
The algorithm determines a  primal solution to the BTS problem such
that $\sum_{ij} c_{ij}f_{ij} \geq (1-\epsilon) OPT$ where $OPT$ is the optimal
solution to the BTS problem.
\end{theorem}
\begin{proof}
The proof is similar to the proof  of  lemma~\ref{lemma_approx}.
The value of the dual solution is 
\begin{eqnarray}
\hspace{-0.25in}&=&\sum_i a_i\alpha_i  + \sum_j b_j\beta_j  + \sum_{ij} u_{ij}\gamma_{ij} \\
\hspace{-0.25in}&=&\sum_i (a_i - \sum_j f_{ij})\alpha_i + \sum_j (b_j - \sum_i p_{ij}f_{ij})\beta_j +
\sum_{ij} [ (u_{ij}-f_{ij})\gamma_{ij}+ 
f_{ij}(\alpha_i+p_{ij}\beta_j+\gamma_{ij})]
\end{eqnarray}

When we subtract the value of the primal solution $\sum_{ij} c_{ij}f_{ij}$, from the
above, the  difference is
\begin{eqnarray}
\sum_i (a_i - \sum_j f_{ij})\alpha_i + \sum_j (b_j - \sum_i p_{ij}f_{ij})\beta_j
-\sum_{ij} f_{ij}(c_{ij} - \alpha_i - p_{ij}\beta_j -\gamma_{ij})  
\end{eqnarray}

Thus, the total absolute difference is at most $\Delta_1 + \Delta_2 + \Delta_3 + \Delta_4$ where

\[
\Delta_1 = |\sum_{ij} f_{ij}(c_{ij} - \alpha_i - p_{ij}\beta_j -\gamma_{ij})|
\]

\[
\Delta_2 = |\sum_{i} \alpha_i(a_{i} - \sum_j f_{ij})|
\]

\[
\Delta_3 = |\sum_{j} \beta_j(b_{j} - \sum_i p_{ij}f_{ij})|
\]

\[
\Delta_4 = |\sum_{ij} \beta_j(u_{ij} - \sum_i \gamma_{ij}f_{ij})|
\]

By arguments similar to lemma~\ref{lemma_approx} and using lemma~\ref{lemma_BTSapprox}
we get that the difference is at most  $\epsilon c_{ij}f_{ij}$.
\ignore{
From lemma \ref{lemma_primalCS}, we have 

\[
\forall ij:f_{ij} > 0,  
\alpha_i \leq c_{ij} - p_{ij}\beta_j +  \epsilon c_{ij}
\]

Therefore, 

\[
\Delta_1 \leq  \epsilon \sum_{ij} c_{ij}f_{ij} 
\]

From the termination condition of the algorithm, we know that for any source
$i$ such that $a_{i} - \sum_j f_{ij} > 0$, we have $\alpha_i = 0$.

Therefore 
\[
\Delta_2 = 0
\]

%by setting $\delta = {\epsilon}{C\sum_i a_i}$ where $C = \max_{ij}c_{ij}$    

For all unsaturated sinks $\beta_j = 0$. Therefore, 
\[
\Delta_3 = 0
\]
}

\end{proof}

\subsection{Complexity}
%The procedures \ref{proc_findpath} \ref{proc_pushpath} and \ref{proc_pushcycle} each takes $O(n)$ time. 
%After finding a path, at least one edge get saturated. Let us charge this work to the saturated
%edge. Let us consider the two cases depending on whether the edge was a forward edge or back edge.
To prove the complexity we  first consider the effect of the pre-processing
step during the algorithm.
In this step the algorithm
removes every 2-cycle unless there is no other back-edge
on the sink $j$.  

The lemma below follows directly:
\begin{lemma} 
\label{lemma_oneloop}
If there is a 2-cycle between a source $i$ and a sink $j$ then there
is no back-edge incident onto  $j$ other than $ij$.
\end{lemma}
The algorithm finds path and cycles and pushes flow on edges
of the path or cycles. Each operation corresponding to a 
push of flow on an edge either enables a reduction
in surplus and changes flow to the capacity of the edge. 
If the flow on a back edge reduces to zero, this  leads to
a change on the dual variables $\beta_j$ for some sink $j$. 
We attempt to charge the various operation in finding paths or cycles and the
subsequent processing to changes in the dual variables.
\begin{lemma}
Each operation in the  procedures  {\bf pushFlowPath, pushFlowCycle, findPath}
can be charged to a rise in the value of $\beta_j, $ for some $j$
such that each rise in the value of $\beta_j$ is charged $O(n^2)$ operations .
\label{lemma_charge}
\end{lemma}
\begin{proof}
We account for the work required for each push of flow along a path or cycle
via a charging argument. 
\ignore{
We will use the following key fact:\\
{\bf Fact 1:} If the flow on edge  $ij$ at valuation
$\beta_j'= \beta_j/(1+\epsilon)$ is reduced to zero, 
then $ij$ occurs as a back edge only when 
%the valuation of $j$ rises
the value of $\beta_j$ rises to $\beta_j(1+\epsilon)$.\\
This fact follows due to the fact the $\beta_j$ is monotonically non-decreasing.
Furthermore, flow can only be pushed on edge $ij$ at an increased valuation 
and an edge
is termed a back-edge when there is flow on the edge that has been pushed at a valuation of $\beta_j/(1+\epsilon)$.
%Flow pushed back on an edge $ij$ reduces the amount of flow pushed into sink $j$ 
%at a valuation of $\beta_j'$. 
}
Procedure {\bf findPath} requires $O(n)$ steps.
Further, during {\bf pushFlowPath}
there are at most $n$ pushes of flow along edges of the path
before one of these happens
\begin{enumerate}
\item
A 2-cycle  is encountered.

If the 2-cycle is between $i$ and $j$, source $i$ increases the
assigned value $y_{ij}$ to $\beta_j$. Since $ji$ was the only back-edge
on $j$ (Lemma~\ref{lemma_oneloop}), there is a change in $\beta_j$. We charge the $O(n)$ pushes 
to this rise in $\beta_j$.

\item
A source with $\alpha=0$ is encountered. This case is analyzed together with the
next case.

\item
An unsaturated sink is encountered.

The previous cases can be considered together with this case as they are similar. In both 
these cases, the surplus from a node travels along the edges of the derived graph
and reaches a source or a sink where no further push of flow is required. The
following sub-cases arise:

\begin{enumerate}

\item
The surplus at the starting node disappears and no new surplus is created

In this case, we charge the  at most $n$ pushes of flow along the edges of the path
to the node at which the surplus 
disappears. We charge such a node only once for the disappearance of the surplus. 
Note that any subsequent
appearance of surplus will be charged to the creation of the surplus and
is caused only in the case when the flow on 
a back-edge becomes zero, which is the second case explained below.  
Considering all possible source-sink pairs,
there are, therefore, at most $n^2$ pushes charged to the nodes
before either the surplus at each source is removed or a surplus appears. 
\ignore{
Appearence of a surplus imples a change in the value
of $\beta$ at the sink $j$ where the flow on the back edge 
We call the sequence of steps between rises in
$\beta$ as a phase. There is a change in some $\beta_j$ at the 
end of the phase. We charge these $n^2$ pushes to this rise
in $\beta_j$. If no $\beta$ rises the algorithm terminates.
}
\item
There is a back edge $ji$ such that $f_{ij}$ becomes zero and a new surplus
at $i'$ is created where $i'j$ was the preferred edge for $i'$.

In this case we charge the push on edges of the path leading to
this back edge, i.e.  a charge of $n$, to
the creation of the surplus at this back edge.
By lemma~\ref{beta-increase} this back edge can not re-appear unless $\beta_j$ changes. We allow an additional
charge of $n$ to pay for the possible disappearance of the flow
from $i'$.

There are at most $n$ back-edges on sink $j$. The value $\beta_j$ changes 
when the flow on each of these edges is reduced to zero. Therefore,
the total charge on all the back-edges incident out of a sink is at most $2n^2$ 
pushes of flow after which there is a change in the value $\beta_j$. This 
implies we have $2n^2$ pushes charged to a rise in $\beta_j$. 

\item
There is a forward edge $ij$ which becomes saturated to 
capacity, i.e. $f_{ij}=u_{ij}$.
If a forward edge is saturated, it stays saturated 
%unless $\beta_j$ increases since a saturated edge
until it becomes a back edge when $\gamma_{ij}=0$. 
When an edge  $ij$ becomes a back edge, by lemma~\ref{beta-increase}
this edge cannot occur again as a back edge  until $\beta_j$ increases.
We can, therefore charge this
work to the corresponding rise in the value of $\beta_j$.
The rise in $\beta_j$ could be charged $O(n^2)$ times as in the
previous case.
Note that as a result of this 
saturating push, surplus 
could be generated on the source just before the edge $ij$. 
We put an additional charge 
of $O(n)$ on this edge to pay for the possible disappearance of this surplus in the future. 

\end{enumerate}

\item
A cycle is encountered. In {\bf pushFlowCycle}, 
all the calculations for determining the limiting number of
revolutions can be done in $O(n)$ time, since the cumulative transfer ratio can be 
determined by traversing the cycle once. As noted above, in all the cases
that arise when flow is pushed around the cycle,
either (i) the surplus that has been pushed into the cycle reduces to 0 %$\delta$, 
or (ii) the flow on one of the back edges is reduced to zero or (iii) the flow
on one of the forward edges reaches capacity. The $O(n)$
amount of work can be charged to either the disappearance 
%(to a negligible amount $\delta$) 
of surplus at a source, or a  reduction of flow to zero on a back-edge $ij$, or
the saturation of a forward edge. This implies  an
increase in  $\beta_j$, subsequently, in the same way as 
described for simple paths above. 

\end{enumerate}

From the case analysis above, we conclude the total charge on a rise
in $\beta_j$ is  $O(n)$ from case 1 and $O(n^2)$ from case 2,3 and 4, hence
proving the lemma.
\end{proof}

\begin{lemma}
The algorithm terminates in $O(\epsilon^{-1}(n^2+ n\log{m})m\log U)$ time, where\\ $U =  
\max_{ij}(\frac{c_{ij}}{p_{ij}}) / \epsilon \min_{ij} (\frac{c_{ij}}{p_{ij}})$.
\end{lemma}

\begin{proof}
From lemma \ref{lemma_charge}, there are $O(n^2)$ charges per rise in  the  
value $\beta_j$. Once a phase is over, the change in value of $\beta_j$ causes sources to 
update the heaps (one at each source). This takes no more than $O(n\log{m})$ time.

The update  procedure takes $O(1)$ amortized time by maintaining a set of
values of  $y_{ij}$,
each value being  either $\beta_j$ or $\beta'_j$.
Thus, we require $O(n\log{m})$ at the end of a phase
to update all the data structures. The preprocessing requires $n$ steps. 

Each rise causes the quantity $\beta_{j}$ to grow by a factor of $1+\epsilon$, starting with
$\epsilon \min_i (\frac{c_{ij}}{p_{ij}})$. 
For each sink $t_j$ such that $\beta_j>0$, there exists an $i$ such that $s_i$ ships flow to $t_j$ and  thus
$0 \leq \alpha_i = c_{ij} - p_{ij}\beta_j$. This  leads to $\beta_j \leq c_{ij}/p_{ij}
\leq \max_{ij}\frac{c_{ij}}{p_{ij}}$. There are $m$ different $\beta_j$. Therefore
there can be no more than $O(m\log_{1+\epsilon} U)$ total changes in  $\beta_j$ for all $j$.  
Combining this with the fact that every rise in $\beta_j$ is charged  $O(n^2+n\log{m})$ amount
of work and $\log_{1+\epsilon}U = O(\epsilon^{-1}\log U)$, we have the result.
\end{proof}

\section{Concave, Piecewise Linear Profit}
\label{piece-wise}

We now describe how to extend the above algorithm to a profit function which 
is concave and piecewise linear. We use the common edge splitting technique 
in order to reduce the problem to the linear profit function. This transformation
is very similar to the well known transformation of convex cost mincost flows to
linear mincost flows, see \cite{ahuja}

Given an instance $I$ of the problem with concave piecewise linear profits,
we map it to an instance $I'$ of the capacitated but linear profit version.

Let the profit function be defined as $c_{ijk} \in \mathbb{Z}$ for the edge $ij$ and interval
$k$ each interval being of fixed length say $l$, total number of intervals 
being ${\cal L}$. 

\begin{figure}[h]
\centerline{\epsfysize=170pt\epsfbox{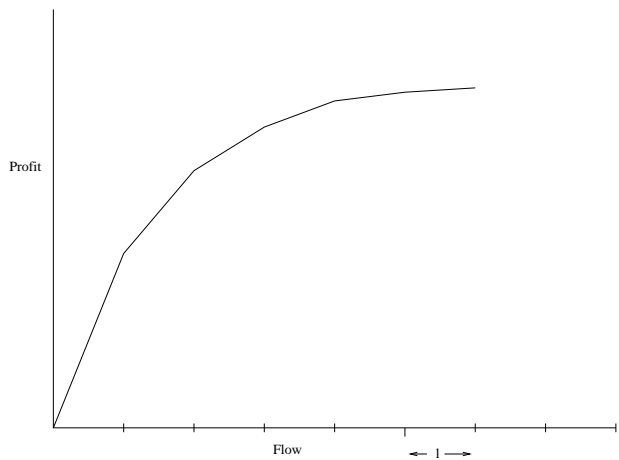}}
\caption{The profit function for an edge $ij$. The slope in an interval $k$ is $c_{ijk}$. Number
of intervals is ${\cal L}$}
\end{figure}

We then convert the problem by  splitting each edge into ${\cal L}$ edges, the
$k$th edge being  of capacity
$l$ and profit $c_{ijk}$. The price of each is the same as the one for edge $ij$. 

\begin{figure}[h]
\centerline{\epsfysize=180pt\epsfbox{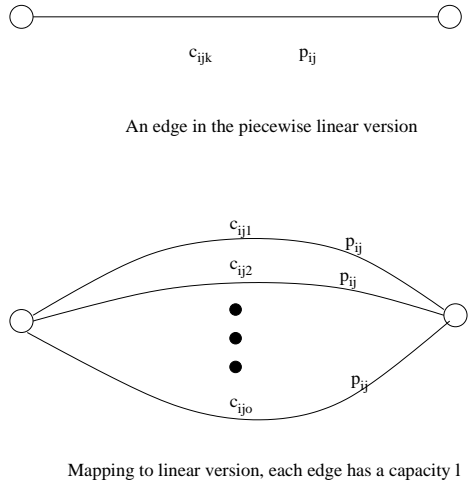}}
\caption{Construction of multiple edges corresponding to one original edge.}
\end{figure}

It is easy to see that a solution to instance $I'$ can be easily converted into a solution for
instance $I'$. A solution to $I'$, however, can not always be converted into a 
solution for $I$. Nevertheless, any solution to $I'$ can always be modified into 
a solution with greater or same profit such that the new solution 
can be transformed into a solution for $I$. 

If for a solution $f_{ijk}$ the following property holds;  $\not \exists z_1, z_2$ such that
 $z_1 < z_2$ and $f_{ijz_1} < l$ and $f_{ijz_2} > 0$, then we can convert the solution to one for $I$. 

If this property does not hold then the profit can be increased/kept the same
while eliminating such pairs. Transfer $\min (l - f_{ijz_1},f_{ijz_2})$
from the edge $ijz_2$ to $ijz_1$. The profit $c_{ijz_1} \geq c_{ijz_2}$ because of
concavity and thus we obtain a more profitable solution. 
Repeating these transfers, we get the property that
 $\not \exists z_1, z_2$ such that
 $z_1 < z_2$ and $f_{ijz_1} < l$ and $f_{ijz_2} > 0$. 

This approach can also be used to map an approximate solution. As a consequence, the
algorithm \ref{alg_modauc} {\bf Modified Auction} can be used to approximate the above problem
after the appropriate transformation. The transformation, however,
does involve an increase in the number of edges. We conclude with the following result.

\begin{theorem}
The Budgeted Transportation problem where the profit function is piecewise-linear and concave,
can be approximated within a factor of $(1-\epsilon)$ in 
$O(\epsilon^-1(n^2+ n\log{m}){\cal L}
m\log U)$ time
where ${\cal L}$ is the number of intervals in the profit function.
\end{theorem}

\section{Conclusion}
\label{sec_conc5}

We have presented an approximation scheme for Budgeted Transportation problem. We further
generalize it to the capacitated version, which can then used to solve the problem
with a non-linear profit version.

The technique used is a kind of primal-dual mechanism based on Auctions. The variables
are modified in small steps in order to maintain approximate slackness conditions. We
use augmenting path based mechanism to improve the complexity of the scheme. However,
because of the nature of the problem, the cycles may arise in addition to simple paths. 
We have shown that the cycles can be handled without too much work. 

An interesting open question is whether it is possible to solve the problem exactly and
as a consequence, generalized flow using the above scheme.

\bibliographystyle{plain}
\bibliography{ref}

\section{Appendix-Relation to Generalized Flow} 
\label{sec_relation}

Apart from being a natural extension of the Transportation Problem BTP is also
related to a well known flow problem, generalized flows. As mentioned above,
it is a special case of Generalized Flow. It is interesting to see, however, that
there  also exists a reverse relationship. We show how to transform the 
generalized flow problem
to Budgeted Transportation.

%{\bf Mincost Generalized Flow : }
\paragraph*{Mincost Generalized Flow: }
We are given a digraph $G(V,A)$, a cost function $c:A\rightarrow \mathbb{R}$, a
capacity function $u:A\rightarrow \mathbb{R}^+$, a multiplier function 
$\mu:A \rightarrow \mathbb{R}^+$ a source $s \in V$, the supply at the source 
$d_s$, a sink $t \in V$ and the demand $d_t$ at the sink. The goal
is to find a flow function, $f:A \rightarrow \mathbb{R}^+$  such that the flow 
is conserved at the nodes,
is multiplied across the arcs and meets the supply and demand constraints, 
while minimizing the total cost of the flow. This problem can be expressed as
the following linear program.

\[
{\rm minimize\ \ \ \ \ }\sum_{ij} c_{ij}f_{ij}
\]

\[
{\rm subject\ to:}\\ 
\]

\begin{eqnarray}
\sum_{i} \mu_{ij}f_{ij} - \sum_{k} f_{jk} &=& 0 \ \ \ \ \forall j \in V/\{s,t\}
\label{eqn_conservation_cons}\\
\sum_{k} f_{sk} &=& d_s\\
\sum_{i} \mu_{it}f_{it} &=& d_t\\ 
f_{ij} &\leq& u_{ij} \ \ \ \ \forall ij \in A
\label{capacity_cons}\\
f_{ij} &\geq& 0 \ \ \ \ \forall ij \in A
\end{eqnarray}

\paragraph*{Mincost Budgeted Transportation: }

This problem is a more generalized version of the conventional 
transportation problem. We are given a bipartite graph $B$, consisting
of sources $S=\{s_i\}$ and sinks $T=\{t_j\}$. A supply function 
$a:S\rightarrow \mathbb{R}^+$, a budget function $b_j \rightarrow  \mathbb{R}^+$, a cost 
function $c:S\times T \rightarrow \mathbb{R}$ and a price function 
$p: S\times T \rightarrow \mathbb{R}$. The goal 
is to come up with a flow function $f:S\times T \rightarrow \mathbb{R}^+$ such
that the total flow going out of a source is equal to the supply,
the total price of the flow coming into a sink is equal to the budget,
and the cost of the flow is minimized. The problem is stated below 
as an LP.

\[
{\rm minimize\ \ \ \ \ }\sum_{i\in S,j \in T} c_{ij}f_{ij}
\]
\[
{\rm subject\ to:}\\ 
\]
\begin{eqnarray}
\sum_{j \in T} f_{ij} = a_i \ \ \ \ i \in S
\label{source_cons}\\
\sum_{i \in S } p_{ij}f_{ij} = b_j \ \ \ \  j \in T
\label{sink_cons}\\
f_{ij} \geq 0 \ \ \ \ i \in S, j \in T
\end{eqnarray}

We first show that an MCGF problem instance  transforms to a minimum Budgeted
Transportation problem instance. The optima of the two instances are the
same and the solutions themselves can be mapped to each other. The transformation
is valid even for approximate solutions.

Let $I$
%:(G(V,A),c,\mu,s,t,d_s,d_t)$ 
be an instance of the MCGF problem. We transform it
to an instance $I'$
%:(S,T,c',p,a,b)$ 
of Mincost Budgeted Transportation problem as follows

\begin{itemize}
\item
For each  node $i$ in $V$, we have a corresponding source
$s_i \in S$ with capacity $a_i = \sum_j u_{ij}$.
%\[
%\forall i \in V, s_i \in S
%\]
\item
For each arc  $(i,j)$ in $A$, we have a corresponding sink in $T$,
$t_{ij}$, with budget $b_{ij}=u_{ij}$.
%\[
%\forall (i,j) \in A, t_{ij} \in T
%\]

\item
We have one additional sink $t_s$ with budget equal to the supply $d_s$ in $I$.
It is connected to the source corresponding to source node $s$ with an
edge whose price is 1 and cost is 0.

\item
Each sink $t_{ij}$ has two incoming edges from $s_i$ and $s_j$, and the
cost and price are as follows
\begin{eqnarray*}
c'_{s_i,t_{ij}} = 0\\
p_{s_i,t_{ij}} = 1\\
c'_{s_j,t_{ij}} = c_{ij}/\mu_{ij}\\
p_{s_j,t_{ij}} = \frac{1}{\mu_{ij}}
\end{eqnarray*}

%\item
%The capacity of source $s_i$ corresponding to node $i$ is $a_i = \sum_j u_{ij}$

\item
The capacity of the source corresponding to the sink in $I$ is equal to
the demand $d_t$

%\item
%The budget of a sink $t_{ij}$ is $b_{ij} = u_{ij}$

\end{itemize}

Consider a  pair  of flow functions $F$ and $F'$ for the instances $I$ and $I'$, respectively.
Let these be related to each other according to the relations defined for 
all $ij$ pairs in the instance $I$.
\begin{eqnarray}
f'_{s_{j},t_{ij}} = f_{ij} \mu_{ij} \\ 
f'_{s_j,t_{jk}} = u_{jk} - f_{jk}
\end{eqnarray}
Further, for the source $s_s$ corresponding to the node $s$,
\[ f'_{s_s, t_s} =  d_s    \]

\begin{lemma}
For a given solution $F$ to the original instance $I$, $F'$ is a feasible solution to
$I'$ of the same cost.
\label{FtoFdash} 
\end{lemma}

\begin{proof}
We show that $F'$ is a feasible solution to $I'$ by showing that it meets the 
constraints \ref{source_cons} and \ref{sink_cons}. 

For each source $s_j$, except the source $s_s$,
the total outgoing flow is the sum of flows
to the sinks corresponding to outgoing and incoming edges on node $j$,
which is,

\begin{eqnarray*}
 &=& \sum_i f'_{s_j,t_{ij}} + \sum_k f'_{s_{j},t_{jk}}\\
\end{eqnarray*}

Using the mapping of solutions, this is equal to 
\begin{eqnarray*}
 &=&  \sum_{i} f_{ij}\mu_{ij} + \sum_k (u_{jk} - f_{jk})\\
 &=& \sum_k u_{jk} + \sum_i f_{ij}\mu_{ij} - \sum_k f_{jk}\\
 &=&  \sum_k u_{jk} = a_j
 \ \ \ \ \ {\rm [using\ equation\ \ref{eqn_conservation_cons}]}
\end{eqnarray*}
Therefore, constraint \ref{source_cons} is met.

For the source $s_s$, a similar analysis shows that constraint
\ref{source_cons} is met,
since $\sum_k f_{sk} = d_s$ and $f'_{s_s , t_s} = d_s   $.

For each sink $t_{ij}$ the total price of the incoming flow is 
\begin{eqnarray*}
f'_{s_i,t_{ij}}/\mu_{ij} + f'_{s_j,t_{ij}}\\
= f_{ij} + u_{ij} - f_{ij} = u_{ij} 
\end{eqnarray*}
Therefore, the constraints \ref{sink_cons} are met.

Also, the total cost of $F'$ 
\begin{eqnarray*}
\sum_{t_{ij}} (0. f'_{s_{i},t_{ij}} +  c'_{s_j,t_{ij}}  f'_{s_{j},t_{ij}})
= \sum_{ij} f_{ij}\mu_{ij} c_{ij}/\mu_{ij} = \sum_{ij} c_{ij} f_{ij}
\end{eqnarray*}
which is same as the cost of $F$.
\end{proof}

\begin{lemma}
For a given solution $F'$ to the  instance $I'$, $F$ is a feasible solution to
$I$ and is of the same cost.
\label{FdashtoF}
\end{lemma}

\begin{proof}
To show that the conservation of flow constraint is met, consider
the total incoming flow on a node $j$, which is 

\begin{eqnarray*}
\sum_i \mu_{ij}f_{ij} = \sum_{i} f'_{s_j,t_{ij}}\\
\end{eqnarray*}

There are two sets of edges incident on source $s_j$. The ones
corresponding to incoming edges on node $j$, $(s_j, t_{ij})$ and
others that corresponding to outgoing edges $(s_j, t_{jk})$. The sum
of  the flow  on these is equal to $a_j$. So the the above is

\begin{eqnarray*}
= a_{j} - \sum_{k} f'_{s_j,t_{jk}} \\
\end{eqnarray*}

Since the total price of incoming flow on sink $t_{jk}$ is equal to
$f'_{s_j,t_{jk}} + \frac{f'_{s_k,t_{jk}}}{\mu_{jk}}$ which is
equal to the budget of the sink $t_{jk} = b_{jk}$, the above becomes,

\begin{eqnarray*}
= a_{j} - \sum_k (b_{jk} - \frac{f'_{s_k,t_{jk}}}{\mu_{jk}})\\
= a_{j} - \sum_k u_{jk} + \sum_k f_{jk}
\ \ \ \ {\rm since\ } b_{jk} = u_{jk}\\
= \sum_k f_{jk} 
\ \ \ \ {\rm since\ } a_{j} = \sum_k u_{jk} 
\end{eqnarray*}
which is equal to the outgoing flow $\sum_k f_{jk}$.
For the case of the source and sink, there are no incoming or outgoing
edges, respectively. But,
\[ \sum_i \mu_{ij}f_{it} = \sum_{i} f'_{s_t,t_{it}} = d_t \]
thus ensuring the sink demand in $I$ and
\[ a_s- d_s = \sum_{k} f'_{s_s,t_{sk}} =  \sum_k u_{sk}-\sum_k f_{sk}.     \]
Thus the outgoing flow at the source,
\[\sum_k f_{sk} =d_s \]
since $a_s = \sum_k u_{sk}$.

Now we show that the capacity constraints are met. For each edge 
$(i,j) \in A$, from constraint 30,  we have
\begin{eqnarray*}
f'_{s_i t_{ij}} + \frac{f'_{s_j t_{ij}}}{\mu_{ij}} = b_{ij} = u_{ij}\\  
\Rightarrow \frac{f'_{s_j,t_{ij}}}{\mu_{ij}} \leq u_{ij}
\ \ \ \ {\rm as\ } f'_{s_i,t_{ij}} \geq 0\\
\Rightarrow  f_{ij} \leq u_{ij}\\
\end{eqnarray*}

The cost of the solutions are mapped exactly as in the proof of lemma 

\ref{FtoFdash}
\end{proof}

Bey combining the lemmas \ref{FtoFdash} and \ref{FdashtoF} we conclude
the following theorem regarding the equivalence of the two problems. 
%\vspace*{.2in}

\begin{theorem}
Given an instance of the MCGF problem $I$ we can construct an instance
of Budgeted Transportation Problem such that the optimum solutions
of each have the same cost.
\end{theorem}

Also, note that the solutions themselves can be mapped to each other.
Thus, by solving the Budgeted Transportation problem we not only determine
value of the optimum solution, we can also construct the solution itself.
\vspace*{.2in}

\subsection{Mincost to Maxprofit Budgeted Transportation}
In order to minimize the cost, we could maximize the negative of the cost
$-c_{ij}$. The maximization, however, does not clear the sources and as such 
conservation of flow in the mapped solution to the original problem would not 
be achieved.

We can get over this problem by minimizing instead, $M - c_{ij}$ where $M$
is a very large constant. Since every edge has a very large profit, 
the optimum will clear all the sources. Any error $\delta$ in clearance will
produce a negative term $M\delta$ in the total profit. By choosing
an $M$ large enough, $\delta$ can be made negligibly small; i.e. smaller
than the granularity of values in the solution of a Linear Program. We refer
to \cite{papadimitriou} for the discussion and bound on this granularity. 

The cost of 
the optimum would be $M\sum_i a_i - \sum_{ij} c_{ij}f_{ij}$, where $\sum_{ij}c_{ij}f_{ij}$ is
an optimum solution for the mincost instance. 
We note that an approximate solution to Maxprofit Budgeted Transportation, would not
immediately lead to an approximation algorithm for Generalized flow.

\end{document}